\documentclass[11pt,runningheads]{llncs}

 \usepackage{amssymb} 
\usepackage{booktabs} 
\usepackage[ruled]{algorithm2e}

\SetAlFnt{\small}
\SetAlCapFnt{\small}
\SetAlCapNameFnt{\small}
\SetAlCapHSkip{0pt}
\IncMargin{-\parindent}

\usepackage{natbib,amsmath,graphicx,amssymb,xcolor}
\usepackage[ruled]{algorithm2e}
\usepackage{geometry} 
\usepackage{caption}
\usepackage{subcaption}

\SetAlFnt{\small}
\SetAlCapFnt{\small}
\SetAlCapNameFnt{\small}
\SetAlCapHSkip{0pt}
\IncMargin{-\parindent}

\newcommand\blfootnote[1]{

  \renewcommand\thefootnote{}\footnote{#1}%
  \addtocounter{footnote}{-1}%

}

\newcommand{\rev}{\text{Rev}}
\newcommand{\prob}{\text{Pr}}
\newcommand{\nind}{\noindent}

\newtheorem{assumption}{Assumption}

	\newcommand{\rej}{\mathcal R}
	\newcommand{\acc}{\mathcal A}
	\newcommand{\Rev}{\text{Rev}}
	
	\newcommand{\sophprob}{\mu}

	\newcommand{\ca}{Property~\ref{prop:a}}
	\newcommand{\cb}{Property~\ref{prop:b}}
	\newcommand{\h}{\mathfrak{h}}

\begin{document}

\title{Repeated Sales with Heterogeneous Buyer Sophistication}

\author{
Rishi Patel \inst{1} 
\and Emmanouil Pountourakis \inst{1} 
\and Sam Taggart \inst{2} 
}

\institute{
Drexel University\\
\email{\{rmp333,ep556\}@drexel.edu}
\and
Oberlin College\\
\email{sam.taggart@oberlin.edu}
}

\maketitle

\begin{abstract}

This paper considers behavior-based price discrimination in the repeated sale of a non-durable good to a single long-lived buyer, by a seller without commitment power. When the buyer fully anticipates the seller's adaptive pricing, the predictions of prior work are bleak: by strategically reducing demand, the buyer can keep prices --- and thus the seller's revenue --- low. In practice, however, buyers vary in their awareness of data tracking, ability to optimize dynamically, and patience. In this work, we therefore study repeated sales with a mixed population of forward-looking ``sophisticated'' buyers and myopic ``naive'' buyers. Buyer heterogeneity introduces additional subtleties to the learning problem, as sophisticated buyers can now feign naivete to potentially secure lower prices. We investigate the impact of these new dynamics on the seller's ability to learn about the buyer and exploit this learning for revenue.

We obtain conclusions that differ dramatically with the time horizon of the interactions. To understand short time horizons, we analyze a two-period model, and find that the strategic demand reduction observed with fully sophisticated buyers is robust to the introduction of naive types. In fact, despite the inability of naive buyers to game the pricing algorithm, their introduction can further harm the seller's revenue, due to more intense demand reduction overall. We prove this phenomenon occurs quite generally. To understand long time horizons, we consider an infinite-horizon model with time discounting. We find that the extreme demand reduction predicted by previous work does not survive the introduction of naive buyers. Instead, we observe equilibria where the seller learns meaningfully despite the sophisticated buyers' demand reduction. We prove that for a natural family of such equilibria, the seller's revenue is not just high, but approximates the revenue attainable with commitment power, even when the fraction of naive types is vanishingly small.

\end{abstract}
\blfootnote{Emmanouil Pountourakis was partially supported by  NSF CCF 2218813. Samuel Taggart was partially  supported by NSF CCF 2218814.}
\pagenumbering{arabic}

\section{Introduction}
\label{sec:intro}

In markets replete with consumer data, {\em behavior-based price discrimination} is commonplace. Users with accounts on online retail platforms have every purchase recorded, along with many non-purchases, in the form of shopping cart data and view tracking. Browser cookies and side-markets for user information further allow for a degree of personalization that carries between platforms. An especially salient consequence of these tracking capabilities is price personalization. This manifests both as overt price discrimination \citep[e.g.][]{streitfeld2000web}, and more obliquely, e.g.\ as coupons and special offers \citep{woolley1998got}. Similar concerns over ratcheted reserve prices also arise in thin, long-lived markets such as auctions for less-demanded online advertising keywords.

The repeated sale of a non-durable good cleanly captures the issues above.
In many applications of personalized pricing, buyers' values are essentially static over time.
A seller who posts prices may try to learn this static value over multiple interactions, e.g.\ by raising a price that has been accepted or lowering a price that was ignored.
When the buyer behaves myopically, a patient seller can learn the buyer's value and extract it as revenue each round \citep{KL03}.

A strategic customer who is aware they are being tracked may change their purchasing behavior to affect future prices.
In particular, such a buyer may engage in {\em strategic demand reduction}, rejecting early prices to ensure lower prices later.
Any pricing algorithm with learning built in can be exploited this way, harming the seller's revenue. A seller faced with strategic buyers may consequently commit not to learn, or may try to learn slowly enough to discourage demand reduction \citep[e.g.][]{ARS13,mohri2014optimal}. 
Conversely, a seller who is unable to commit to a pricing strategy can paradoxically pay for their ability to learn. 
Such a seller adapts their pricing based on the buyer's behavior while the buyer similarly adapts to deceive the seller.
The resulting equilibria are well-studied, and the common prediction is significant revenue loss, well below what the seller would obtain with a static price \citep{ht88,s93,ilpt17,dps19}.
Intuitively, a strategic buyer can play the seller at present off against the seller in the future, ``bargaining'' for lower prices.

This paper departs from previous work by considering a buyer population with heterogeneous levels of strategic sophistication. 
That is, we assume some buyers are {\em strategically sophisticated}, and plan their purchasing behavior with future prices and utility in mind, while other buyers are {\em strategically naive} (i.e.\ myopic), responding to the current prices with no regard for their future rounds of interaction with the seller.
This heterogeneity could stem from a variety of causes: some buyers may be more aware of data-tracking technology than others; buyers may vary in their ability to optimize, as is the case in ad markets that pit small businesses against large companies with tailored bidding software; or it may just be that buyers have varying patience levels (while still all transacting with the same seller multiple times). 
These practical scenarios lie in between the high-revenue, all-naive world and the low-revenue world of exclusively strategic buyers.

We assume the seller lacks commitment power, and is unable to directly observe the buyer's level of sophistication.
Under these assumptions, a heterogeneous population  exacerbates the seller's difficulties with inference.
Strategic demand reduction becomes even more convincing: a buyer who rejects a price could now be naive, and could truly have a low value. 
On the other hand, because naive buyers won't strategize, they make more willing customers.
Depending on the balance of buyer types in the population, the seller may be differently swayed to tune their prices to one group versus the other.
Predicting equilibrium outcomes involves teasing out these new dynamics.

\subsection{Model and Contributions}

The canonical model for behavior-based price discrimination without commitment is due to \citet{ht88}, which we adjust to introduce heterogeneous buyer sophistication.
At the beginning of the game, a buyer arrives with private, persistent value $v$ and sophistication level $\theta\in\{N,S\}$ drawn jointly from a known distribution.
Each period $t$, the seller offers a price $p_t$, which the buyer either accepts or rejects. 
Naive buyers always accept if and only if $v\geq p_t$.
Sophisticated buyers play a perfect Bayesian equilibrium (PBE) with the seller, to be defined formally in Section~\ref{sec:prelims}. 
Informally, PBE requires both the seller and the sophisticated buyers to best respond sequentially to each others' strategies, combined with beliefs obtained from Bayesian updating.
In other words, the sophisticated buyers maximize their total utility across the game, while the seller's prices in each round must maximize revenue against the {\em mixture} of naive buyers and sophisticated buyers.

We analyze equilibria in two variants of this model: one with a short time horizon and few interactions, and the other with a long horizon and many.
For short time horizons, we consider a two-period model.
For long horizons, we consider an infinite-horizon model, with utilities geometrically discounted by patience parameter $\delta$.
We focus on the high-patience regime ($\delta\geq 1/2$).
In each model, we compare to prior work with fully-sophisticated buyers.
We find that robustness of existing results is quite different between the two models, which we detail below.

\paragraph{Results: Short Time Horizon}

In the two-period model, we assume the buyer's value $v$ and sophistication $\theta\in\{N,S\}$ are drawn independently, with value distribution $F$ and probability of sophistication $\sophprob$.
Our analysis proceeds in two stages.
We begin by fully solving the linear-demand case where $v\sim U[0,1]$, and analyze the equilibrium structure and revenue as the sophistication level $\sophprob$ varies.
The first observation is that equilibrium undergoes a phase transition.
At low values of population sophistication ($\sophprob\leq \overline \mu\approx .63$), the seller targets naive buyers with a descending price schdule and obtains high revenue.
In this ``naive-focused'' regime, the seller's revenue decreases with the introduction of sophisticated buyers, who engage in demand reduction.
At high sophistication levels ($\sophprob\geq \overline \mu$), the seller switches to an ascending price schedule, which targets sophisticated buyers and yields lower revenue.
This ``sophisticated-focused'' regime resembles the fully-sophisticated equilibrium derived by \citet{dps19}, and shows their analysis to be robust.
In fact, these low-revenue predictions are more than robust: our second finding is that adding naive buyers can further decrease the seller's expected revenue below the already-low fully-sophisticated baseline.

We then show that the two main observations from the linear-demand case hold much more generally: for all value distributions satisfying a mild regularity condition, we prove that there exists a neighborhood about full sophistication with ascending prices and low revenue.
Moreover, in this region, the seller's revenue is strictly increasing in the sophistication level $\mu$, despite the tendency of sophisticated buyers to game the seller's pricing algorithm.
We conclude the two-period analysis with several smaller results that confirm that this revenue nonmonotonicity stems from intensified demand reduction, paired with the lack of seller commitment.

\paragraph{Results: Long Time Horizon}
Next, we consider the infinite-horizon, geometrically-discounted model.
With full sophistication and $\delta\geq 1/2$, \citet{dps19} observe the presence of an equilibrium with extreme demand reduction and no price search: when the value distribution is supported at zero, the equilibrium predicts that the seller offers the item for free each day.
Despite the extreme nature of this ``no-learning'' prediction, \citet{ilpt17} further show this equilibrium to uniquely survive a set of ad hoc refinements.
We begin our analysis of the heterogeneous setting with two straightforward observations.
First, the no-learning equilibrium does not survive the introduction of naive buyers.
Second, it is possible to construct other natural equilibria in the presence of naive buyers, which we illustrate for a discretely distributed example.

Motivated by these observations, we go on to prove a broad revenue guarantee for equilibria with heterogeneous buyers:
any equilibrium satisfying two mild properties --- properties which hold for our example --- guarantees the seller revenue which approximates that attainable {\em with commitment power}. 
For any $\delta$ bounded away from $1$, our theorem implies a constant-approximation, significantly improving on the zero-revenue prediction for $\delta\geq 1/2$ with full sophistication.
Notably, the revenue approximation does not depend on the probability the buyer is naive.
In fact, even without naive buyers, both the motivating example and revenue result survive, though the equilibrium fails the refinements of \citet{ilpt17}.
This suggests a competing prediction to the no-learning equilibrium of \citet{dps19} which is both practically-motivated and less dire.
The results for this setting hold even when the buyer's sophistication is correlated with their value.

\subsection{Related Work}

%
This paper extends a long line of work in economics and adjacent areas.
The model of repeated pricing with no commitment studied in this paper originates with \citet{ht88} and \citet{s93}, and later continues with \citet{dps19}, \citet{ilpt17}.
The prevailing conclusion from these papers is that with long time horizons, the seller cannot hope for more than trivial revenue.
This observation is closely related to the literature on durable goods monopoly, initiated by \citet{coase1972durability}, and which notably includes \citet{stokey1981rational}, \citet{gul1986foundations}, \citet{ausubel1989reputation,ausubel1992durable}, \citet{bagnoli1989durable}, and \citet{sobel1983multistage}.
Much of this work focuses on {\em Coasian dynamics,} where the seller competes with themselves in future time periods, resulting in low prices and low revenue. 
Our revenue result for long time horizons suggests that the main intuition from both lines of work is not robust in the practical situation where the buyer population is heterogeneous.

Our short-horizon results, meanwhile, stand in contrast to the work on pricing mechanisms for strategic buyers, including \citet{ARS13}, \citet{drutsa2017horizon,drutsa2018weakly,drutsa2020reserve}, \citet{mohri2014optimal}, and \citet{zhiyanov2020bisection} in computer science, and e.g.\ \citet{board2016revenue} in economics.
The universal starting point for this work is the observation that strategic consumers damage the seller's bottom line.
Our two-period results offer a more nuanced view: with buyer heterogeneity, we show that additional sophisticated buyers can actually improve the seller's revenue.

Many other papers explore related themes in similar models.
\citet{CTW12}, \citet{ali2020voluntary}, \citet{cummings2016strange}, and \citet{li2020transparency} all consider privacy, transparency, or consumer control over their information.
The phenomena we consider are also related to the {\em ratchet effect} \citep{freixas1985planning, baron1984regulation, laffont1988dynamics, weitzman1980ratchet}, where a regulator's choice to update their targets based on past performance introduces dynamic incentive issues for a firm.
Finally, the marketing literature has devoted much attention to the way behavior-based price discrimination affects inter-firm competition in duopolistic models. See e.g.\ \citet{pazgal2008behavior,bimpikis2021data,zhang2011perils,li2016behavior}, as well as the survey of \citet{villas2004price}.

Mechanism design with heterogeneous behavioral types is common to three lines of literature. 
The first is the literature on mechanism design for agents with private, heterogeneous levels of patience.
Notable contributions include \citet{mierendorff2016optimal}, \citet{burkett2021intertemporal}, \citet{pai2013optimal}, and \citet{deb2015dynamic}.
    This work assumes the seller has commitment power.
    Our seller lacks commitment, but otherwise our model can be thought of under this umbrella: a naive buyer is equivalent to a buyer with extremely low patience (and therefore myopic behavior), and a sophisticated buyer is far-looking and patient.
    An interesting future direction is understanding how broadly our results extend into models with less dichotomous patience levels.

Two recent papers explicitly consider Bayesian mechanism design with a combination of forward-looking and myopic buyers, and observe increasing revenue as the number of sophisticated buyers grows \citep{ADMS18,agrawal2019dynamic}.
Two details of their setting renders this observation unsurprising: first, the seller has commmitment power.
Second, they consider a setting where the buyer's value is drawn fresh each period.
Under these circumstances, standard results \citep[e.g.][]{ashlagi2016sequential} show how to extract nearly full surplus from forward-looking buyers using the fact that the sum of their future values is concentrated.
As observed in the literature on pricing to strategic buyers, a persistent buyer value makes similar results impossible, and instead introduces dynamic incentive problems.

Our work is loosely connected to the literature on reputation \citep{kreps1982reputation,kreps1982rational,milgrom1982predation,abreu2007bargaining,watson1996reputation,cramton1992strategic}. 
This work differs from ours in that the high level goal is to enforce particular equilibria of repeated games.
Behavioral types are introduced as a way to give a strategic player a way to bluff credibly as they bargain.
Despite the differing high-level goal, we observe a similar phenomenon.
When our seller loses revenue with the introduction of naive types, it is precisely because these types allow the buyer greater credibility as they bluff low values.
\section{Notation and Preliminaries}
\label{sec:prelims}

\paragraph{Game Description} The timing of the game is as follows.
At time $0$, a buyer is drawn from the population.
A buyer is described by a private value $v$ and a private sophistication level $\theta\in\{S,N\}$ (with $S$ denoting ``sophisticated'' and $N$ ``naive'').
The value $v$ and sophistication $\theta$ are drawn jointly from a prior which is common knowledge.
We think of first drawing $\theta$ with $\mu=\text{Pr}[\theta=S]$, then drawing $v$ from distribution $F^\theta$, for conditional distributions $F^N$ and $F^S$.
At each time $t$, the seller selects a price $p_t$, and the buyer makes a decision $D_t\in\{0,1\}$, indicating whether they wish to purchase the item or not.
The game ends after $T$ rounds: for our purposes; we consider $T=2$ for the short-horizon results and $T=\infty$ for the long-horizon results.
Both players discount their utilities by a factor $\delta\in [0,1]$.
Given a sequence of prices and decisions, utilities are the following:
\begin{align*}
    U_{\text{seller}}(\{p_t\}_{t=1}^T,\{D_t\}_{t=1}^T)=\sum\nolimits_t \delta^{t-1}p_tD_t&& U_{\text{buyer}}(\{p_t\}_{t=1}^T,\{D_t\}_{t=1}^T)=\sum\nolimits_t \delta^{t-1}(v-p_t)D_t.
\end{align*}
Note that the value $v$ is not redrawn between rounds.

\paragraph{Naive Buyers.} When the buyer is naive, they behave as a price taker. 
That is, for a buyer with $\theta = N$, $D_t=\mathbb I(v\geq p_t)$ for $t\in \{0,1\}$.
As naive buyers are nonstrategic, all subsequent discussion of equilibrium should be understood to apply only to sophisticated buyers.

\paragraph{Histories}
A history of play for a buyer describes all private information and past actions that are observable to that player.
For our game, we have:
\begin{itemize}
    \item A round $t$ buyer history $h_t^B$ consists of the buyer's value $v$, the current-round price $p_t$, and all previous prices and decisions, $(p_1,D_1,\ldots,p_{t-1},D_{t-1})$.
    \item A round $t$ seller history $h_t^S$ consists of all previous prices and decisions, $(p_1,D_1,\ldots,p_{t-1},D_{t-1})$.
\end{itemize}

\paragraph{Equilibrium}
Our solution concept is perfect Bayesian equilibrium (PBE).
A PBE consists of strategies $\boldsymbol\sigma$ and beliefs for each period.
For player $\rho\in \{S,B\}$, and time $t$, the strategy $\sigma_t^\rho$ maps the history for $\rho$ in round $t$ to an action in that round.
For the seller, this is a price, and for a sophisticated buyer, this is an accept/reject decision.
Beliefs map histories of play to distributions over private information.
For our purposes, it suffices to consider only the seller's beliefs, as the buyer observes all private information.
In particular, the seller's beliefs in each round will be a distribution over $(v,\theta)$ pairs.
Perfect Bayesian equilibrium imposes two requirements:
\begin{itemize}
    \item Bayesian updating: the seller's beliefs are derived from Bayes' rule whenever possible, given the strategies.
    Beliefs for off-path histories may be arbitrary.
     \item Sequential rationality: Given a history $h_t^\rho$ for player $\rho$, the strategy $\sigma_t^\rho(h_t^\rho)$ maximizes $\rho$'s expected utility given the other player's strategies and $\rho$'s beliefs at $h_t^\rho$.
\end{itemize}

\paragraph{Threshold Equilibria} We impose an additional refinement that selects simple PBE.
We assume that sophisticated buyers' strategies are monotone, in the sense that higher values are more likely to buy in a given round.
In particular, this implies that the buyer's strategy can be summarized by a {\em threshold}: given a round-$j$ buyer history $h_j^B$, there is some $t_j(h_j^B)$ such that $\sigma_j^B(h_j^B)=1$ if and only if $v\geq t_j(h_j^B)$.
Previous work has focused on threshold equilibria as well \citep{dps19,ilpt17}.

\subsection{Short-Horizon Setting: Notational Simplifications}

The short-horizon analysis focuses on the special case of $T=2$ and $\delta=1$, and we further assume for that analysis that the value $v$ and sophistication are drawn independently, so $F^N=F^S=F$.
The seller's initial beliefs are therefore given by $\mu$ and $F$.
Under these assumptions, we can make some additional simplifications and define some further helpful notation.

In general, the seller's beliefs after round $1$ would consist of a correlated joint posterior over $\mu$ and $v$.
With two rounds, however, the seller's second-round decision is a simple monopoly pricing problem, as both naive and sophisticated buyers will act as price takers.
Consequently, it suffices to collapse the second-period beliefs over $\mu$ and $v$ to a posterior belief over $v$.
Moreover, price-taking behavior is trivially a threshold strategy, with $t_2(h_2^B)=p_2$, so we omit subscripts when discussing $t_1$, as it is the only nontrivial threshold.
Given these observations, the following second-period notation is helpful: for truncation point $x$, let $F_{\leq x}(v)=\min(F(v)/F(x),1)$ and $F_{\geq x}(v)=\max((F(v)-F(x))/(1-F(x)),0)$ denote the CDFs of the truncated distributions. Let $\mu$, $p_1$, and $t=t(p_1)$ be the population sophistication, first-round price, and first-round threshold. If the seller sees a rejection, their posterior over $v$ is given by
\begin{equation}
\label{eq:frej}
    F^\rej(v;p_1,t)=\mu_\rej F_{\leq t}(v)+(1-\mu_\rej) F_{\leq p_1}(v),
\end{equation}
where
\begin{align*}
    \mu_\rej=\tfrac{\mu F(t)}{\mu F(t)+(1-\mu)F(p_1)}&&1-\mu_\rej=\tfrac{(1-\mu) F(p_1)}{\mu F(t)+(1-\mu)F(p_1)}.
\end{align*}
Similarly, the seller's posterior on a round-$1$ accept is given by
\begin{equation}
\label{eq:facc}
    F^\acc(v;p_1,t)=\mu_\acc F_{\geq t}(v)+(1-\mu_\acc) F_{\geq p_1}(v),
\end{equation}
where
\begin{align*}
    \mu_\acc=\tfrac{\mu (1-F(t))}{\mu (1-F(t))+(1-\mu)(1-F(p_1))}&&1-\mu_\acc=\tfrac{(1-\mu) (1-F(p_1))}{\mu (1-F(t))+(1-\mu)(1-F(p_1))}.
\end{align*}
The formulas $F^\rej$, $F^\acc$, $\mu^\rej$, and $\mu^\acc$ depend $\mu$, $t$, and $p_1$, which at various points in our arguments will be fixed or varied as appropriate.
We will consequently write these as functions of subsets of these parameters depending on our need.
Context will make dependencies clear.

In round $2$, the seller maximizes their revenue function $R^D(p)=p(1-F^D(p))$ for $D\in\{\rej,\acc\}$.
Further notation will aid in discussing revenue maximization. 
Let $R(p)=p(1-F(p))$ denote the unupdated revenue curve. 
We will use $p^*$ to denote the revenue-maximizing {\em monopoly price}, and $R^*=R(p^*)$ its revenue.
Similarly, we can let $R_{\leq x}(p)=p(1-F_{\leq x}(p))$ and $R_{\geq x}p(1-F_{\geq x}(p))$ denote the revenue curves of truncated distributions, let $p_{\leq x}^*$ and $p_{\geq x}^*$ denote their optimizers, and $R_{\leq x}^*$ and $R_{\geq x}^*$ denote their optimal values.
\section{Short-Horizon Warmup: Equilibrium With Linear Demand}
\label{sec:uniform}

In this section, we give a detailed analysis of equilibrium for the linear-demand case where $v\sim U[0,1]$, for all values of the sophistication parameter $\mu$. 
In particular, we characterize the prices and the buyer's threshold strategy {\em on the equilibrium path}. 
For equilibrium to be well-defined, it is also necessary to define strategies at histories that are off-path --- in this case, off-path choices of first-round price $p_1$.
Due to the extensive casework involved, we defer presentation of the full strategy functions to Appendix~\ref{app:uniform}. 
In that appendix, we also give more detailed analysis of the equilibrium conditions.
In Section~\ref{sec:general}, we will generalize many of our insights far beyond the linear-demand case.

\begin{figure}
     \centering
     \begin{subfigure}[b]{0.45\textwidth}
         \centering
         \includegraphics[scale=0.6]{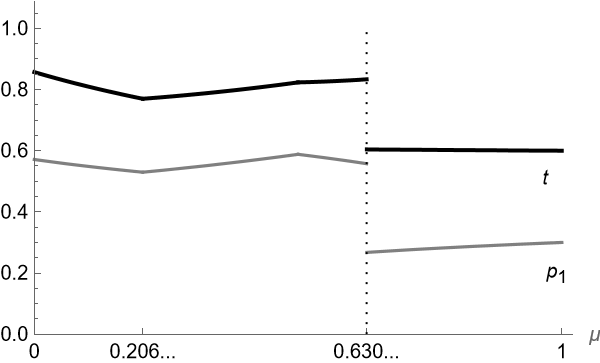}
         \caption{First-round price $p_1$ (gray) and first-round threshold $t$ (black)\\
         \\}
         \label{fig:plot1a}
     \end{subfigure}
     \hfill
     \begin{subfigure}[b]{0.45\textwidth}
         \centering
         \includegraphics[scale=0.6]{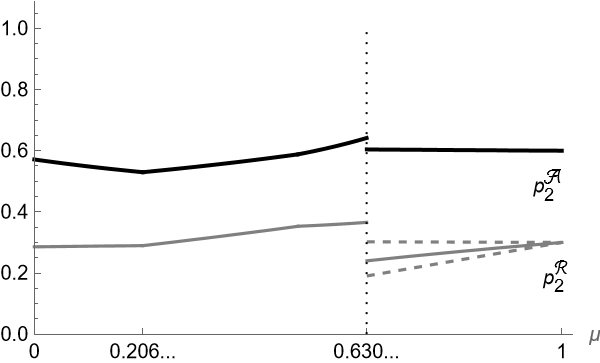}
         \caption{Second round prices on accept (black) and reject (gray). For $\mu\geq 0.630...$, dotted gray lines indicate the support of $p_2^\rej$, and solid gray line the expected value.}
         \label{fig:plot1b}
     \end{subfigure}
     \caption{On-path prices and threshold for the linear-demand case.}
     \label{fig:plot1}
\end{figure}

Figure~\ref{fig:plot1} shows the structure of equilibrium on the path of play. 
Figure~\ref{fig:plot1a} gives the seller's optimal choice of $p_1$ as a function of the population sophistication $\mu$, along with the buyer's threshold response for that choice of $p_1$.
In round $2$, the seller offers a price $p_2^\acc$ on seeing a round-$1$ accept, and offers a possibly randomized price $p_2^\rej$ on seeing a round-$1$ reject.
Figure~\ref{fig:plot1b} shows the value of $p_2^\acc$ as a function of $\mu$, as well as the two values in the support of the seller's mixed strategy conditioned on round-$1$ rejection.
The precise probabilities, along with the functional forms for all the prices, can be found in Appendix~\ref{app:uniform}.
Note the absence of a second-round buyer threshold: the buyer acts as a price-taker in round $2$, i.e.\ the second-round threshold is always $p_2$.
The presented price paths are only for the on-path choice of $p_1$, which maximizes the seller's total revenue.
Were the seller to choose a different, off-path first-round price, $t$, $p_2^\acc$, and $p_2^\rej$ would change in response.

For all values of $\mu$, we note two common structural details. 
First, note that with probability $1$, $p_2^\acc>p_2^\rej$.
The round $2$ distribution conditioned on accept is a mixture of naive buyers with value in $[p_1,1]$ and sophisticated buyers with value in $[t,1]$.
Conditioned on reject, the distribution mixes naive buyers with value in $[0,p_1)$ and sophisticated buyers with value in $[0,t)$.
The second-round prices are optimal for these conditional distributions (or tied for optimal, in the case of $p_2^\rej$), and the stronger distribution (accept) in this case has a higher monopoly price.
Second, note that $t> p_1$, i.e.\ the sophisticated buyers always employ nontrivial demand reduction. 
The threshold value solves the buyer indifference equation $t-p_1+(t-p_2^\acc)^+=t-\mathbb E[p_2^\rej]$; the marginal type gets equal utility from accepting (lefthand side) and rejecting (righthand side).
If $p_2^\acc< t$, then $p_1<t$ follows from the fact that $p_2^\acc>p_2^\rej$. 
If $p_2^\acc\geq t$, then $p_1=\mathbb E[p_2^\rej]$. 
Since all realizations of $p_2^\rej$ must be optimal for the second-round distribution conditioned on reject, and this latter distribution is supported on $[0,t]$, it must be that $p_1=\mathbb E[p_2^\rej]< t$.


\subsection{Equilibrium Structure}
\label{sec:structure}

Equilibrium can be broken into two regimes, based on the level of population sophistication $\mu$. 
For low values of $\mu$, the population comprises mostly naive buyers. 
The seller focuses their pricing on these buyers, and the price schedule remains similar to the all-naive case of $\mu=0$.
We refer to this as the {\em naive-focused} regime.
For high values of $\mu$, the seller must cater to the sophisticated buyers.
The price structure now resembles the $\mu=1$ case, but $p_2^\rej$ is randomized for reasons we discuss below.
This is the {\em sophisticated-focused} regime: the first main result of the next section is that such a neighborhood always exists about full sophistication, and hence the conclusions of \citet{dps19} are robust to buyer heterogeneity under short horizons.

\paragraph{Naive-Focused Regime} 
For $\mu<0.630$, the seller chooses prices adapted to the large number of naive buyers.
The price $p_1$ starts high, and on an accept, stays high; for $\mu\leq 1/2$, $p_1=p_2^\acc$.
On seeing a reject, it is likely that a naive buyer with value in $[0,p_1]$ is responsible, and consequently $p_2^\rej<p_1$.
The sophisticated buyers get away with significant demand reduction, with $t>p_2^\acc$.
When $\mu\in[0.5,0.630...]$, the large number of sophisticated buyers with $v\geq t$ causes $p_2^\acc$ to rise even above $p_1$, but otherwise the structure of equilibrium is as just described.

\begin{figure}
     \centering
     \begin{subfigure}[b]{0.3\textwidth}
         \centering
         \includegraphics[scale=0.5]{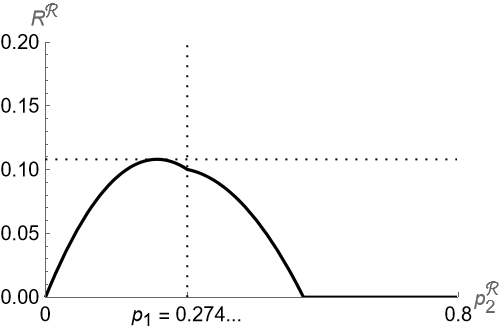}
         \caption{$t$ too low}
         \label{fig:plot2a}
     \end{subfigure}
     \hfill
     \begin{subfigure}[b]{0.3\textwidth}
         \centering
         \includegraphics[scale=0.5]{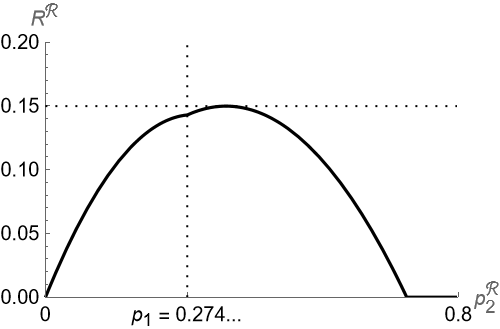}
         \caption{$t$ too high}
         \label{fig:plot2b}
     \end{subfigure}
     \hfill
     \begin{subfigure}[b]{0.3\textwidth}
         \centering
         \includegraphics[scale=0.5]{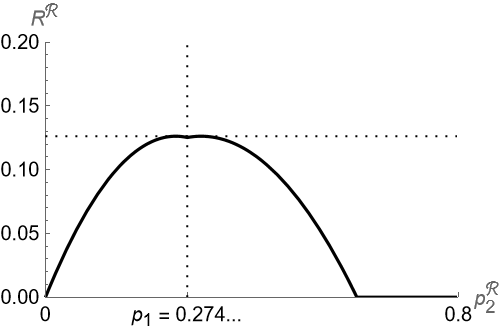}
         \caption{$t$ admits mixed $p_2^\rej$}
         \label{fig:plot2c}
     \end{subfigure}
     \caption{Illustration of the seller's pricing problem after first-round reject for $\mu=0.7$. The only way to satisfy the threshold equation $p_1=\mathbb E[p_2^\rej]$ with a sequentially rational price is to balance the two optima of $R^\rej$. Vertical dotted line: $p_1$. Horizontal dotted line: optimal second-round revenue.}
     \label{fig:plot2}
\end{figure}

\paragraph{Sophisticated-Focused Regime}
At $\mu\approx0.630$, all prices jump discontinuously.
The equilibrium now resembles the $\mu=1$ equilibrium derived in \citet{dps19}, with the additional complication that $p_2^\rej$ is now randomized.
Prices now start low, and because an accept likely comes from a sophisticated buyer with $v\in[t,1]$, we have $p_2^\acc=t$.
The sophisticated buyers' indifference equation becomes $t-\mathbb E[p_2^\rej]=t-p_1$, and so $\mathbb E[p_2^\rej]=p_1$; a reject does not significantly change the price between the first and second rounds, except to induce the seller to randomize in round $2$.

Two factors together render randomization of $p_2^\rej$ necessary. 
First, for $\mu<1$, it cannot be that $p_2^\rej$ is deterministically $p_1$, no matter how $t$ is chosen.
In more detail, when $t$ is significantly higher than $p_1$, relatively few buyers have values below $p_1$.
The best price in this case is the monopoly price for $U[p_1,t]$, which for $t\geq p_1$ is $t/2>p_1$.
When $t$ is low and close to $p_1$, the best price lies in $[0,p_1)$.
In between the high and low $t$ cases, there are local maxima in both $[0,p_1)$ and $(p_1,t]$, both of which yield better revenue than $p_1$.
The second issue is that the threshold equation $t-\mathbb E[p_2^\rej]=t-p_1$ implies $\mathbb E[p_2^\rej]=p_1$.
Randomization then resolves these two issues by choosing $t$ such that the peaks in $[0,p_1)$ and $(p_1,t]$ yield the same revenue, allowing the seller to randomize between then while satisfying sequential rationality.
Figure~\ref{fig:plot2} illustrates the seller's second-round pricing problem conditioned on a reject.
Note that the indifference induced by this randomization is different from what is typical in the construction of mixed-strategy equilibria: rather than inducing the indifference necessary for the buyer to in turn play a mixed strategy of their own, the seller's randomization induces an indifferent buyer type necessary for the threshold equilibrium refinement.

\subsection{Revenue Analysis}
\label{sec:revuniform}

Figure~\ref{fig:plot3} illustrates the seller's equilibrium revenue. 
At the all-naive $\mu=0$ extreme, the seller can adaptively price without dynamic incentive issues, and obtain revenue $4/7\approx0.571$.
They can therefore beat the benchmark revenue of $0.5$, achieved by offering the monopoly price $p^*=0.5$ both rounds.
At the all-sophisticated $\mu=1$ extreme, the \citet{dps19} PBE revenue $0.45$, lower than the monopoly pricing benchmark.
In between the two extremes, the seller's equilibrium revenue is given by:
\begin{equation*}
    \Rev(\mu)=
    \begin{cases}
        \tfrac{(2+\mu)^2}{7+10\mu+3\mu^3}&0\leq \mu < 0.206...\\
        \tfrac{9+2\mu-5\mu^2-2\mu^3}{(4+\mu-\mu^2)^2}&0.206...\leq \mu <0.5\\
        \tfrac{-7-2\mu+5\mu^2+2\mu^3}{-12-8\mu+6\mu^2+5\mu^3+\mu^4}&0.5\leq \mu < 0.630...\\
        \tfrac{(1+2\sqrt\mu)^2}{4+8\sqrt\mu+7\mu+2\mu^{3/2}-\mu^2}&0.630...\leq \mu \leq 1,
    \end{cases} 
\end{equation*}
with the precise values of the truncated decimals given in Appendix~\ref{app:uniform}.
We can better understand the revenue based on the two equilibrium regimes described in the previous section.

\begin{figure}
    \centering
    \includegraphics[scale=0.6]{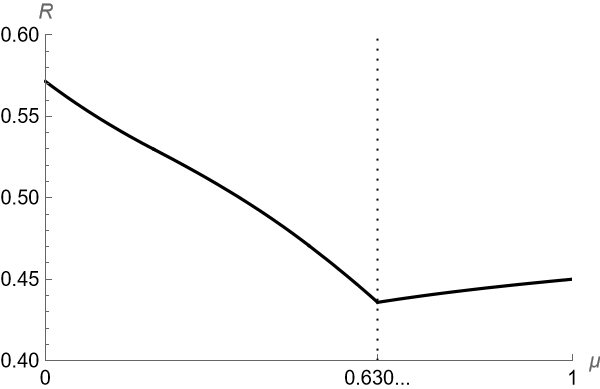}
    \caption{Plot of the seller's equilibrium revenue.}
    \label{fig:plot3}
\end{figure}

\paragraph{Naive-Focused Regime} 
For $\mu<0.630...$, the seller's revenue is decreasing in the population sophistication level.
This loss can be attributed to two sources.
First, sophisticated buyers engage in strategic demand reduction, whereas naive buyers are price-takers.
Consequently, the average revenue from a sophisticated buyer is generally less than the average revenue from a naive buyer; increasing $\mu$ replaces the latter type with the former type.
This is {\em population-driven revenue loss}.
A second source of loss comes from the way the equilibrium structure shifts as $\mu$ increases.
The seller must gradually adapt their prices to handle the new sophisticated buyers, and this reduces their ability to exploit the naive buyers with adaptive pricing.
This is {\em price-driven revenue loss}.

\paragraph{Sophisticated-Focused Regime}
For $\mu\geq 0.630...$, the seller's revenue is {\em increasing} in the population sophistication level.
Again, we may consider population-driven and price-driven changes to the revenue.
As in the naive-focused regime, the seller faces a population-driven revenue loss, as sophisticated buyers continue to generate less revenue than naive buyers.
In the sophisticated-focused regime, however, the seller's price schedule is structured around the large number of sophisticated buyers.
In this case, adding sophisticated buyers results in a {\em price-driven revenue gain}, as the presence of naive buyers distorts the price structure in a way that reduces the revenue from sophisticated buyers.
The figure shows that the price-driven effects are significant enough to overshadow the population-driven effects.
Below, we show that the revenue increase with added sophistication is a general phenomenon.
The second main result of the next section is that this phenomenon persists for every distribution satisfying some standard regularity conditions.
\section{Short Horizon: General Analysis}
\label{sec:general}

We now generalize two of the main conclusions of the linear-demand analysis to a broader family of value distributions.
The first main result of the section is that the price structure of the full-sophistication equilibrium is robust in a neighborhood about $\sophprob=1$.
That is, for a region of sufficiently high $\mu$, equilibrium has the same ascending-price structure observed in the sophisticated-focused regime of the previous section.
In this region, the seller's revenue is no greater than the low revenue of the full-sophisticated equilibrium.
In fact, the second main result of the section implies that it is strictly lower: we show that in this same neighborhood, the seller's revenue is increasing in the population sophistication level $\mu$.
Thus, the nonmonotonicity of revenue observed in Section~\ref{sec:revuniform} is a general phenomenon.

For the rest of this section, we maintain the assumptions that the value distribution $F$ is atomless, fully supported on $[0,1]$, and satisfies the following concavity condition {\em strictly}. We omit these assumptions in our theorem statements for brevity.
\begin{assumption}\label{ass:regular}
The value distribution $F$ has a price-posting revenue curve $R(p)=p(1-F(p))$ which is strictly concave (i.e.\ $R''(p)<0$) on $[0,1]$.
\end{assumption}
It is possible to relax our distributional assumptions and preserve many of our main conclusions, albeit at the cost of expositional clarity.

The structure of this section is as follows.
Sections~\ref{sec:2basics} and \ref{sec:2t} contain basic definitions and useful facts.
Section~\ref{sec:2basics} focuses on the structure of prices in equilibrium.
Section~\ref{sec:2t} characterizes the relationship between these prices and buyer's threshold in more detail.
In Section~\ref{sec:2robust}, we state and prove the main results.


\subsection{Basic Structural Results}
\label{sec:2basics}

This section gives basic preliminary results, and defines terminology needed to discuss the structure of equilibrium.
In equilibrium, the seller picks a first-round price $p_1$ which maximizes their revenue, given the buyer's response.
Even for the same $\mu$, different choices of $p_1$ will induce different price structures which resemble the two types of equilibrium observed for the linear-demand case.
The following mild abuse of terminology will allow us to distinguish between the seller's options:

\begin{definition}
\label{def:continuation}
    A {\em continuation} is a quadruple $(p_1,p_2^\rej,p_2^\acc,t)$, with $p_2^\rej$ and $p_2^\acc$ possibly random, satisfying the following three conditions:
    \begin{enumerate}
        \item\label{item:ropt} For all $p$ in the support of $p_2^\rej$, $p\in\arg\max p(1-F^\rej(p;p_1,t))$.
        \item\label{item:aopt} For all $p$ in the support of $p_2^\acc$, $p\in\arg\max p(1-F^\acc(p;p_1,t))$.
        \item\label{item:threshold} $t$ satisfies the indifference equation $t-p_1+\mathbb E[(t-p_2^\acc)^+]=\mathbb E[(t-p_2^\rej)^+]$.
    \end{enumerate}
    A continuation with deterministic $p_2^\acc$ is {\em naive-focused} if $p_2^\acc<t$, and {\em sophisticated-focused} if $p_2^\acc\geq t$.
    An equilibrium is naive-focused or sophisticated-focused if the on-path choice of $(p_1,p_2^\rej,p_2^\acc,t)$ is respectively naive-focused or sophisticated-focused.
\end{definition}

Note that for a fixed $p_1$, there could in principle be multiple choices of $(p_2^\rej,p_2^\acc,t)$ that form a valid continuation. 
Similarly, we can think about starting with any subset of $\{p_1,p_2^\rej,p_2^\acc,t\}$, and filling in the remaining quantities to form a continuation. Formally:

\begin{definition}
    Given choices for a subset $S\subseteq\{p_1,p_2^\rej,p_2^\acc,t\}$, we say a choice of the other parameters $\{p_1,p_2^\rej,p_2^\acc,t\}\setminus S$ {\em implements} the choices for $S$ if together those parameters form a continuation.
    An implementation is {\em naive-focused} if the continuation formed is naive-focused, and {\em sophisticated-focused} if the continuation is sophisticated-focused.
\end{definition}

Equilibrium requires every first-round price $p_1$ to have a valid continuation.
For the sake of completeness, we show prove the following in Appendix~\ref{app:exists}.
\begin{lemma}\label{lem:p1implement}
Every first-round price $p_1$ has a valid implementation.
\end{lemma}
The proof generalizes the discussion of mixed strategies for the naive-focused regime of Section~\ref{sec:uniform}. If the seller's optimal second-round prices changed continuously with the threshold $t$, then the utility of a sophisticated buyer with that threshold type would also change continuously, and we would be guaranteed that this buyer's accept and reject utilities would be equal for some $t$.
When the second-round prices can jump, such a crossing point may not exist.
However, mixing between the prices on either side of that jump produces buyer indifference.
This immediately implies that equilibria exist, which we also prove in Appendix~\ref{app:exists}.

\begin{theorem}\label{thm:exists}
For any distribution $F$ and sophistication level $\mu$, there exists a (possibly mixed) perfect Bayesian equilibrium for the two-round game.
\end{theorem}

We next give basic structural properties of equilibrium prices.
Conditions~\ref{item:ropt}, \ref{item:aopt}, and \ref{item:threshold} in Definition~\ref{def:continuation} together imply strong conditions on the price structure, which we summarize here and prove in Appendix~\ref{app:ordering}.

\begin{lemma}\label{lem:ordering}
    Let $F$ have monopoly reserve $p^*$. Then for any continuation $(p_1, p_2^\rej,p_2^\acc,t)$ with $p_1>0$, the following inequalities hold for every realization of $p_2^\rej$ and $p_2^\acc$:
    \begin{enumerate}
        \item\label{item:p2rlimit} $p_2^\rej \leq t$
        \item\label{item:p2alimit} $p_2^\acc \geq p_1$
        \item\label{item:p2rmonop} $p_2^\rej \leq p^*$
        \item\label{item:p2amonop} $p_2^\acc \geq p^*$
        \item\label{item:threshorder} $p_1\leq t$
    \end{enumerate}
    We may consequently write the buyer's indifference equation as: $t-p_1+\mathbb E[(t-p_2^\acc)^+]=t-\mathbb E[p_2^\rej]$ .
\end{lemma}

In the case where the continuation is sophisticated-focused, we may say even more about the second-round prices.
Of special note is the reject price $p_2^\rej$.
As we observed in the linear-demand case, the only way to make the threshold buyer type indifferent may be to randomize $p_2^\rej$.
Our second lemma below (Lemma~\ref{lem:plph}), below shows this to always be case under Assumption~\ref{ass:regular}.

\begin{lemma}\label{lem:plph}
	Let $(p_1,p_2^\rej,p_2^\acc,t)$ be a sophisticated-focused continuation. Then $p_2^\rej$ is supported on exactly two values, $p_L$ and $p_H$, satisfying:
	\begin{itemize}
		\item $p_H$ is the monopoly price for the truncated distribution $F_{\leq t}$.
		\item $p_L$ maximizes $\mu_\rej R_{\leq t}(p)+(1-\mu_\rej)R_{\leq p_1}(p)$ on $[0,p_1]$.
	\end{itemize}

Furthermore, if $\mu=1$ then $p^\rej_2=p_L=p_H$ is deterministic. If $\mu<1$, then $p^\rej_2$ is randomized.  
\end{lemma}

The proof can be found in Appendix~\ref{app:mixture}.
We can also characterize the accept price $p_2^\acc$.
The following lemma is proved in Appendix~\ref{app:a2det}.

\begin{lemma}
\label{lem:p2adet}
    For any continuation $(p_1,p_2^\rej,p_2^\acc,t)$, the accept price $p_2^\acc$ is deterministic.
    If the continuation is sophisticated-focused, then $p_2^\acc=\max(t,p^*)$.
\end{lemma}

\subsection{Sophisticated-Focused Thresholds}
\label{sec:2t}

Before stating and proving the main results, we present three lemmas at the core of the technical approach.
Both main results require an understanding of sophisticated-focused continuations and their revenue.
Rather than analyzing the revenue and price structure in terms of the seller's choice of first-round price $p_1$, the analysis will fix the buyer's threshold $t$ and let the corresponding first-round price $p_1$ vary with $\mu$.
The following lemma, proved in Appendix~\ref{app:ttop} makes this possible.

\begin{lemma}\label{lem:ttop}
    Given $t$ and $\mu$, $(1-\mu)R'(t)+(1-F(t))\mu\geq 0$ if and only if $t$ has a sophisticated-focused implementation.
    This implementation is unique.
\end{lemma}

Fixing $\mu$, call a $t$ satisfying $(1-\mu)R'(t)+(1-F(t))\mu\geq 0$ {\em sophisticated-focused}.
Lemma~\ref{lem:ttop} has three important consequences for such $t$.
First, note that the derivative of $(1-\mu)R'(t)+(1-F(t))\mu$ with respect to $\mu$ is 
\begin{equation*}
    -R'(t)+(1-F(t))=-(1-F(t)-tf(t))+(1-F(t))\geq 0.
\end{equation*}
Hence, if $(1-\mu)R'(t)+(1-F(t))\mu\geq 0$ holds, it also holds for all $\mu'>\mu$.
In other words, if a threshold $t$ is sophisticated-focused at $\mu$, it is also sophisticated-focused for all larger $\mu$.
Similarly, the derivative of $(1-\mu)R'(t)+(1-F(t))\mu$ with respect to $t$ is negative, since both $R'(t)$ and $1-F(t)$ are decreasing with respect to $t$, so if $t$ is sophisticated-focused for some $\mu$, so too is every $t'<t$.
Finally, note that Lemma~\ref{lem:ttop} implies that as long as $(1-\mu)R'(t)+(1-F(t))\mu\geq 0$, we can define the function $p_1(t,\mu)$ to be the unique $p_1$ implementing $t$.
The following lemma shows that $p_1(t,\mu)$ is well-behaved.

\begin{lemma}\label{lem:p1monotone}
    For $t$ and $\mu$ satisfying $(1-\mu)R'(t)+(1-F(t))\mu\geq 0$, the first-round price $p_1(t,\mu)$ is differentiable and strictly increasing in $\mu$.
\end{lemma}

We prove Lemma~\ref{lem:p1monotone} in Appendix~\ref{app:p1monotone}.
Finally, we show that any $p_1$ that has a sophisticated-focused implementation cannot have a naive-focused implementation as well.
Hence, for the first-round prices $p_1$ supporting sophisticated-focused continuations, their corresponding threshold is unique: there is a bijection between such prices and thresholds.
We prove this result in Appendix~\ref{app:p1nonaive}.

\begin{lemma}\label{lem:p1nonaive}
    Any $p_1$ with a sophisticated-focused implementation does not admit a naive-focused implementation.
\end{lemma}

\subsection{Main Results}
\label{sec:2robust}

We now state our main conclusions for the two-round analysis.
First, for sufficiently high $\mu$, the structure of equilibrium matches that of the sophisticated-focused regime described in the linear-demand analysis.

\begin{theorem}\label{thm:robust}
    For every distribution $F$, there exists some $\mu_F<1$ such that for all $\mu>\mu_F$, equilibrium is sophisticated-focused.
\end{theorem}

Note that the structural results of Section~\ref{sec:2basics} pin down the price structure in sophisticated-focused equilibrium.
Hence, for $\mu>\mu_F$, we have $p_1=\mathbb E[p_2^\rej]$ and $p_2^\acc=\max(t,p^*)$, by Lemmas~\ref{lem:plph} and \ref{lem:p2adet}, respectively.
The second main result of the two-round analysis is that in the sophisticated-focused region, adding sophisticated buyers actually improves the seller's revenue.

\begin{theorem}\label{thm:2rev}
    Let $\mu_F$ be as in Theorem~\ref{thm:robust}.
    For all $\mu>\mu_F$, the seller's equilibrium revenue is strictly increasing for $\mu\in(\mu_F,1]$. That is, for any $\mu_F<\mu<\mu'\leq 1$, the seller's revenue is strictly greater in any PBE for $\mu'$ than for $\mu$.
\end{theorem}

The main insight driving both results is that we may fix any sophisticated threshold $t$ and analyze its revenue as $\mu$ varies.
In more detail, note that Lemma~\ref{lem:ttop} implies that whenever a threshold $t$ has a sophisticated-focused implementation, this implementation is unique, and $t$ continues to have a sophisticated-focused implementation for any higher $\mu$.
We therefore fix $t$ and vary $\mu$, analyzing the revenue of this implementation as its prices vary with $\mu$.

\begin{lemma}\label{lem:trev}
    Fix any threshold $t$, and let $\mu$ be such that $t$ is implementable in a sophisticated-focused continuation.
    The revenue of this continuation is continuous and strictly increasing in $[\mu,1]$.
 \end{lemma}

The challenge in proving Lemma~\ref{lem:trev} is that per capita, naive buyers actually generate more revenue than sophisticated buyers.
The seller's revenue gains are instead driven by shifts in the price structure, towards lower-revenue prices.
Fixing $t$ simplifies the analysis of this price shift: by Lemma~\ref{lem:p2adet}, $p_2^\acc=\max(t,p^*)$ does not change with $\mu$.
Moreover, Lemma~\ref{lem:plph} implies that $p_2^\rej$ is supported on two prices that yield equal revenue for the seller.
One of these prices is $p^*_{\leq t}$, which again does not change with $\mu$ once $t$ is fixed.
Hence, the only only price that requires more detailed analysis is $p_1$.
But $p_1$ is increasing, by Lemma~\ref{lem:p1monotone}.
Accounting for the revenue carefully in terms of these three prices yields the result; the detailed proof can be found in Appendix~\ref{app:trev}.

To prove Theorem~\ref{thm:robust}, then, we show that for high enough $\mu$, the there is a sophisticated-focused threshold that dominates all naive-focused thresholds. 
In more detail, for high-enough $\mu$, the seller can emulate the fully-sophisticated equilibrium (which has a sophisticated-focused continuation) and attain nearly the same revenue.
Meanwhile, Lemma~\ref{lem:ttop} implies that only very high thresholds can be naive-focused.
Intuitively, these thresholds imply that the seller earns very little and revenue gains very little information in period $1$, and are thus low-revenue overall.
Hence, the sophisticated-focused thresholds are preferable.
The full proof can be found in Appendix~\ref{app:robust}
Given Theorem~\ref{thm:robust}, it is straightforward to prove Theorem~\ref{thm:2rev}.

\begin{proof}[Theorem~\ref{thm:2rev}]
    Consider $\mu$ and $\mu'$ satisfying $\mu_F<\mu<\mu'\leq 1$.
    Let $t$ be the on-path threshold for $\mu$.
    By Lemma~\ref{lem:ttop}, $t$ also has a (unique) sophisticated-focused implementation under $\mu'$.
    By Lemma~\ref{lem:trev}, the revenue from this implementation of $t$ under $\mu'$ is higher than that under $\mu$, i.e.\ $\rev(\mu,t)< \rev(\mu',t)$.
    Moreover, under $\mu'$, the first-round price $p_1(\mu',t)$ implementing $t$ has no other implementation, and hence the seller could attain $\rev(\mu',t)$.
    The seller picks the revenue-optimal choice of $p_1$.
    Letting $\rev(\mu)$ and $\rev(\mu')$ denote the PBE revenues, we have:
    \begin{equation*}
        \rev(\mu)=\rev(\mu,t)<\rev(\mu',t)\leq \rev(\mu').
    \end{equation*}
\end{proof}

\subsection{Comparative Analyses}
In Appendix~\ref{app:comparative}, we give three analyses aimed at at more clearly identifying the cause of the nonmonotonicity observed in Theorem~\ref{thm:2rev}.
First, we confirm the intuition that naive buyers are per capita more lucrative for the seller.
Consequently, it cannot be that revenue increases with $\mu$ because of population-driven revenue effects.
Second, we show that for the linear-demand case, the welfare in equilibrium {\em decreases} with the level of sophistication.
Put another way, the introduction of naive buyers causes the Coasian dynamics of the full-sophistication equilibrium to grow even more Coasian as the sophisticated buyers gain further bargaining power.
This also rules out an alternative explanation of the revenue gain, namely that the revenue increases because the overall welfare increases.
Finally, we show that the nonmonotonicity is tied to the seller's lack of commitment.
We analyze the same two-period model in the presence of seller commitment and show that in this alternative setting, adding sophisticated buyers only harms the seller's revenue.

\section{Long Horizon Analysis}

We now turn to long time horizons, where the buyer and seller transact for many time periods.
In particular, we analyze an infinite horizon model with time discounting and high discount factor $\delta$.\footnote{An alternative would be to consider a model with a long but finite horizon, with little or no discounting. We opt for the infinite-horizon setting because of its relative analytical tractability. Even without naive buyers, analyses of finite horizons require additional strong assumptions on the values \citep{ht88} or the discount factor \citep{gavious2017price}.}
This model is well-studied \citep{dps19,ilpt17}, and previous work has focused on one particular equilibrium, termed the {\em no-learning equilibrium}.
In the no-learning equilibrium, which exists for $\delta\geq 1/2$, the seller offers the lower bound of the buyer's support every period.
All buyer types accept this price every period, leading the seller's beliefs to remain static.
The seller's revenue in this equilibrium is trivial: in particular, when the lower bound of the buyer's support is $0$, the seller gives the item away for free every period.
This equilibrium is enforced by punishments for off-path behavior: any deviation price is rejected by all buyers, and failing to reject a such a price leads to prohibitively high prices in the future.
A full account of the equilibrium is included in Appendix~\ref{app:nolearning} for completeness.
Despite the counterintuitive nature of this equilibrium, \citet{ilpt17} show that the no-learning equilibrium is the unique survivor of a set of ad hoc equilibrium refinements.

The no-learning equilibrium is not robust to the introduction of naive buyers.
This is best illustrated when the buyer's value distribution is supported on some interval $[0,H]$.
In this case, the no-learning equilibrium dictates that the seller offer a price of $0$ for the entire game, earning revenue $0$.
However, as long as there is a positive mass of naive buyers with positive value --- say, at least $v$ --- the seller can earn strictly positive revenue by posting $v$ each round.
Hence, posting $0$ every round is no longer sequentially rational, no matter the strategy of the sophisticated buyers.
This discussion proves the following theorem:
\begin{theorem}
	For the infinite-horizon setting with $\sophprob<1$, there is no PBE where the seller offers price $0$ every period.
	In particular, the no-learning strategies in Appendix~\ref{app:nolearning} are not an equilibrium.
\end{theorem}

This section characterizes equilibria in the presence of naive buyers.
The central observation is that there exist natural equilibria that are qualitatively distinct from the full-sophistication no-learning equilibrium.
In Section~\ref{sec:infeq}, we solve for equilibrium in an example, to illustrate the impact of naive buyers.
The equilibrium has three notable features.
First, the structure of equilibrium is simple, and robust to perturbations in parameters.
Second, the seller successfully learns about the sophisticated buyers, and obtains nontrivial revenue.
Finally, the successful learning and nontrivial revenue hold {\em even when the fraction of naive buyers is vanishingly small}.
The equilibrium does not converge to the no-learning equilibrium as the probability of naive buyers goes to zero.

In Section~\ref{sec:infrev}, we provide a more general revenue analysis of infinite-horizon PBE.
We identify two features of the equilibrium from Section~\ref{sec:infeq}, and show that these features are sufficient for the seller to obtain high revenue.
In particular, we show that the seller's revenue approximates the maximum revenue attainable {\em with commitment}.
The approximation ratio depends on the discount factor $\delta$, and for any $\delta$ bounded away from $1$, the approximation ratio is constant.
Notably, the approximation ratio does not depend on the probability that the buyer is naive.

We conclude our analysis of the long horizon setting by observing that the equilibrium of Section~\ref{sec:infeq} persists even with zero naive buyers.
This suggests an alternative equilibrium refinement for the fully-sophisticated regime which eliminates the no-learning equilibrium and instead admits other equilibria with high revenue.
Given that the no-learning equilibrium is non-robust and viewed as counterintuitive, this refinement merits further study.
We elaborate in Section~\ref{sec:infrefine}.

\subsection{Motivating Example}
\label{sec:infeq}

In this section, we demonstrate via an example that the introduction of even a small mass of naive buyers can lead to qualitatively different dynamics in equilibrium, compared to the no-learning equilibrium of the fully-sophisticated regime.
For the example, we consider a discrete distribution: we assume both naive and sophisticated values are uniformly distributed over $V=\{1,10,20\}$.
To simplify the calculations and to emphasize the non-robustness of the all-sophisticated equilibrium, we take the mass $(1-\sophprob)$ of naive buyers to be some small $\epsilon$.
We take discount factor $\delta=2/3$.
The equilibrium structure we obtain will be robust to perturbations of all the parameters above.
Note that even though values are supported on a finite set, we allow the seller to offer any nonnegative price, including those outside $V$.

Before describing the equilibrium, it is helpful to make a few comparisons.
First, note that if the seller had commitment power, the optimal fixed price would be tied between $10$ and $20$, each yielding expected per-round revenue of $20/3$ and total discounted revenue $20/3\cdot1/(1-\delta)=20$.
In other words, the seller would be best served targeting the high-value buyers.
Meanwhile, if the seller had no commitment power, and if the naive buyers were not present, the no-learning equilibrium dictates that the seller to offer $1$ every round, for total revenue $1/(1-\delta)=3$.

In the equilibrium we describe below, the seller targets the high-value sophisticated buyers, albeit less effectively than if they had commitment power.
We give the on-path behavior of the seller and the sophisticated buyers.
That is, for each round and on-path history, we give the seller's optimal price and the buyer's threshold response.
To define the full equilibrium and to verify that agents are best responding, it is necessary to specify strategies for all histories, including those that are off the path of play.
We do so in Appendix~\ref{app:infinite3sup}.

\paragraph{Equilibrium Description.} The on-path equilibrium behavior is simple to describe.
In period $1$, the seller offers $p_1=2$.
Sophisticated buyers with value $20$ accept, those with value $10$ and $1$ reject (i.e.\ the threshold is $20$).
The behavior in subsequent periods depends on the period $1$ decision.
If $p_1$ was accepted, the seller offers $p_2^\acc=10$ perpetually, which is accepted by all remaining buyers.
If $p_1$ was rejected, the seller offers $p_2^\rej=1$ perpetually, which again is accepted by all remaining buyers.
The seller's expected discounted revenue (neglecting the $\epsilon$ mass of naive buyers) is
\begin{equation*}
	\tfrac{1}{3}\cdot(p_1+\tfrac{\delta}{1-\delta}\cdot p_2^\acc)+\tfrac{2}{3}\cdot\tfrac{\delta}{1-\delta}\cdot p_2^\rej=\tfrac{26}{3}=8.66\ldots
\end{equation*}

\paragraph{Incentive Analysis.}
We briefly sketch the incentive analysis, starting with period $2$.
After seeing a period $1$ accept, the seller's beliefs are supported on $\{10,20\}$ for the naive buyers, and $\{20\}$ for sophisticated buyers.
It this continuation game, incentives resemble those of the fully-sophisticated no-learning equilibrium.
The sophisticated buyers reject any price above $10$.
The seller must choose between offering $10$, which is accepted by all remaining types, or offering $20$, which is only accepted by naive buyers with value $20$.
With sufficiently few naive buyers, $p_2^\acc=10$ is optimal.
The continuation following a period $1$ reject has a similar structure: the support of naive buyers is $\{1\}$, and of sophisticated buyers is $\{1,10\}$.
All prices above $1$ are rejected by sophisticated buyers.

Given the period $2$ strategies, we can analyze period $1$.
A sophisticated buyer with value $20$ is indifferent between accepting and rejecting in period $1$.
To see this, note that the utility of such a buyer for accepting is $U^\acc=20-p_1+\tfrac{\delta}{1-\delta}(20-p_2^\acc)=38$ and for rejecting is $U^\rej=\tfrac{\delta}{1-\delta}(20-p_2^\rej)=38$.
Finally, the seller's period $1$ price is chosen to be the highest $p_1$ capable of inducing this buyer response:
Any higher price leads to a strict incentive for value-$20$ buyers to reject.
Lower prices induce the value-$20$ buyers to accept, but are dominated by $p_1=2$.

\paragraph{Off-Path Analysis: Notable Details.}
A full analysis of the strategies which includes off-path histories is given in Appendix~\ref{app:infinite3sup}.
For the time being, we highlight two structural elements that are present in the example and will be shown in the next section to imply a general revenue guarantee.

First, we note that for any history, the seller never prices above of below the support of the naive buyers remaining in the distribution.
To see that the seller never prices above the current upper bound $\h_t^N$ of the naive support, note that, any price $p> \h_t^N$ is rejected by all sophisticated buyers.
Furthermore, accepting such a $p>\h_t^N$ is off-path, so on such an accept, we update the beliefs to full sophistication with value $20$.
The seller then posts a price of $20$ for the rest of the game.
Since any price $p> \h_t^N$ will be rejected, such prices are dominated by the optimal $p\leq \h_t^N$.
Meanwhile, rejecting $p>\h_t^N$ is sequentially rational for the sophisticated buyers.
If they accept, they get utility $v-p< v-\h_t^N$.
Meanwhile, if they reject, all future prices will be at most $\h_t^N$ by the previous argument, and hence their utility will be at least $\tfrac{\delta}{(1-\delta)}(v-\h_t^N)$, which is at least $v-p$ when $\delta\geq 1/2$.
The seller similarly doesn't price below the naive support because all such prices are accepted, which we discuss below.

The second noteworthy structural element of the example equilibrium is that starting from any history, the seller is guaranteed at least $\ell_t^N/(1-\delta)$ in revenue, where $\ell_t^N$ is the lower bound of the naive buyers.
This follows because a price of $\ell_t^N$ is always accepted.
This behavior is sequentially rational, but this least obvious for period $1$, which we argue here.
(Given the full strategy descriptions in the appendix, this is simple to verify at all other histories.)
First, consider the utility of a sophisticated buyer with value $v$ who rejects a price $p_1\leq \ell_1^N=1$.
Such an action is off-path: we choose to update the seller's beliefs to the lowest sophisticated type in this case, i.e.\ $1$, and price at this value for all future rounds.
As a result, the buyer's utility for rejecting is $U^\rej(v)=\tfrac{\delta}{1-\delta}(v-1)=2v-2$.
Meanwhile, accepting leads the seller not to update their beliefs next round, and to commence the on-path price sequence by offering a price of $2$.
The buyer's utility from this price sequence is
\begin{equation*}
	U^\acc(v)=v-p_1+\delta\max(v-2+\tfrac{\delta}{1-\delta}(v-10),\tfrac{\delta}{1-\delta}(v-1)),
\end{equation*}
with the two cases in the max coming from the buyer's choice when faced with the period $2$ price of $2$.
It is straightforward to verify that $U^\acc(v)\geq U^\rej(v)$ for all $v\geq 1$.
This observation has two further consequences.
First, it implies that all prices strictly below $\ell_t^N$ are dominated and will not be offered in equilibrium.
Second, it implies a lower bound of $\ell_t^N/(1-\delta)$ on the seller's revenue after any history, as they could always price at the lower bound of the naive support in perpetuity.
In fact, it is exactly this pricing behavior that leads to the bulk of the seller's revenue from the on-path price sequence, starting in period $2$.

\subsection{Revenue Guarantee}
\label{sec:infrev}

The previous section gave an illustrative equilibrium for the infinite horizon model with heterogeneous buyers.
We observed that the equilibrium satisfied the two properties below:

\begin{property}[Naive-Justified Prices]\label{prop:a}
	After any seller history $h^S_t$ inducing beliefs supported on some naive type, the seller's equilibrium price $p_t(h^S_t)$ satisfies $\ell^N_t\leq p_t\leq \h^N_t$, where $\ell^N_t$ and $\h^N_t$ are the respective lower and upper bounds of the support for the naive buyer's value after $h_t^S$.
\end{property}

\begin{property}[Revenue Above Baseline]\label{prop:b}
	After any seller history $h^S_t$ inducing beliefs supported on some naive type, the seller's total expected discounted revenue starting from $h^S_t$ is at least $\ell^N_t/(1-\delta)$,  where $\ell^N_t$ is the lower bound of the support for the naive buyer's value after $h_t^S$.
\end{property}

Both \ca\ and \cb\ should be thought of as mild requirements.
The example equilibrium of Section~\ref{sec:infeq} demonstrates how these conditions can be satisfied by equilibria with natural structures, obtained via straightforward backward induction.
The main result of this section, informally, is that any threshold PBE satisfying \ca\ and \cb\ has high revenue. 
The result will hold for all continuous distributions, as well as for all discrete distributions with a sufficiently ``dense'' grid, to be formalized shortly.
Moreover, the main result of this section holds even if naive and sophisticated buyers have different value distributions (or equivalently, buyers' values are correlated with their level of sophistication).
To this end, define $F^S$ and $F^N$ to be the distributions of the value $v$ conditioned on the buyer being sophisticated and naive, respectively, and let $F$ denote the unconditional value distribution.
We assume throughout that $F^S$ and $F^N$ have the same support.
As in the rest of the paper, let $\sophprob$ denote the probability the buyer is sophisticated.

As a revenue benchmark, we consider the maximum expected discounted revenue the seller could attain if they had commitment power.
To upper bound this quantity, first consider the event $\mathcal E_N$ where the buyer is naive.
If the seller knew they buyer was naive, they could solve a Markov decision process (MDP) to obtain an optimal dynamic pricing policy.
Let $\rev^N$ denote the revenue of this policy against the naive buyers with distribution $F^N$.
Meanwhile, in the event $\mathcal E_S$ that the buyer was sophisticated, then a straightforward simulation argument shows that the seller's maximum revenue is achieved by posting a static price every round, in particular the monopoly price $p^*_S$ for $F^S$ (\citet{dps19}, Proposition 1).
This yields revenue $\sophprob \max_p p(1-F^S_+(p))/(1-\delta)$, where we write $F_+^S(x)=\prob[v<x~|~\mathcal E_S]$ to account for the fact that we break ties in favor of the seller.
(We otherwise write $F^S(x)=\prob[v\leq x~|~\mathcal E_S]$ without the subscript to denote the standard right-continuous CDF.)
Overall, we derive the following upper bound on the commitment benchmark, obtained by noting that the seller's commitment revenue given the knowledge of the buyer's sophistication is at most that without this knowledge:
\begin{equation}\label{eq:benchmark}
	\sophprob \max_p p(1-F^S_+(p))/(1-\delta)+(1-\sophprob)\rev^N.
\end{equation}

To characterize equilibrium revenue, we first note two simple lower bounds on equilibrium revenue.
First, for any equilibrium satisfying \cb\ with support lower bound $\ell=\inf V$, it holds that the seller's revenue is at least $\ell/(1-\delta)$.
Second, note that naive buyers will behave the same no matter the strategy of the sophisticated buyers.
Hence, the seller could choose to ignore the sophisticated buyers entirely, and instead play the optimal MDP-based pricing strategy for the naive buyers.
Hence:
\begin{lemma}\label{lem:trivialrev}
	For any threshold PBE satisfying \cb\, the seller's expected discounted revenue is at least $\max(\ell/(1-\delta),\rev^N)$.
\end{lemma}

Lemma~\ref{lem:trivialrev} implies that difficult term in (\ref{eq:benchmark}) to approximate is the first, representing the revenue from sophisticated buyers.
To approximate the sophisticated revenue, we now state the main result of the section:

\begin{theorem}\label{thm:revenue}
For any threshold PBE with naive buyers satisfying \ca\ and \cb, the seller's expected discounted revenue is at least $\max_{p\in V}\tfrac{\delta}{1-\delta}p(1-F(p/(1-\delta)))$.
\end{theorem}

Two points of discussion are in order before proving Theorem~\ref{thm:revenue}.
First, Theorem~\ref{thm:revenue} implies that the with the introduction of naive buyers, any equilibria satisfying \ca\ and \cb\ exhibit fundamentally different dynamics than the no-learning equilibrium of the fully-sophisticated regime.
In particular, these equilibria enable the seller to learn and collect revenue, potentially significantly above the baseline of $\inf V/(1-\delta)$ attainable in the no-learning equilibrium.
Notably, Theorem~\ref{thm:revenue} holds no matter how small the mass of naive buyers is, and the revenue guarantee itself has no dependence on the fraction of the population that is naive.
This rules out a continuous transition to the no-learning equilibrium as the naive buyers vanish.

Second, we can compare the revenue guaranteed by Theorem~\ref{thm:revenue} to the sophisticated portion of the revenue benchmark (\ref{eq:benchmark}), given by $\sophprob \max_p p(1-F^S_+(p))/(1-\delta)$.
For values of $\delta$ bounded away from $0$, we show below that the revenue guarantee of Theorem~\ref{thm:revenue} compares favorably to this benchmark.
If the value distributions are continuous, then the comparison is simple to state.

\begin{lemma}\label{lem:contreg}
	For any continuous regular distribution supported on $[0,H]$, 
	$$\max_{p\in V}p(1-F( p/(1-\delta)))\geq (1-\delta)\sophprob\max_{p\in V}p(1-F^S_+(p)).$$
\end{lemma}
\begin{proof}
	Let $p^*_S$ denote the monopoly price for $F^S$, maximizing $p(1-F^S_+(p))$.
	Consider taking $p=(1-\delta)p^*_S$, and let $\mathcal E_S$ denote the event that the buyer is sophisticated.
	Then we have:
	\begin{align*}
		\max_{p\in V}p(1-F( p/(1-\delta)))&\geq (1-\delta)p^*_S(1-F(p^*_S))\\
		&\geq (1-\delta)p^*_S\prob[v>p^*_S]\\
		&\geq (1-\delta)p^*_S\prob[\mathcal E_S]\prob[v>p^*_S~|~\mathcal E_S]\\
		&=(1-\delta)\sophprob p^*_S(1-F^S(p^*))\\
		&=(1-\delta)\sophprob p^*_S(1-F^S_+(p^*)).
	\end{align*}
	The final line follows from the continuity of the distribution, and the rest from definitions or basic probability.
\end{proof}

Lemma~\ref{lem:contreg} implies a $\delta(1-\delta)$-approximation to the sophisticated piece of the benchmark, $\sophprob \max_p p(1-F^S_+(p))/(1-\delta)$.
Without naive buyers, the no learning, trivial-revenue equilibrium exists for all $\delta\geq 1/2$, so for all $\delta$ bounded away from $1$, this marks a significant improvement.
For discrete distributions, the analog of Lemma~\ref{lem:contreg} requires more care to state, but the conceptual conclusions are the same: for any distributions with sufficiently dense mass points, the revenue bound of Theorem~\ref{thm:revenue} is near the sophisticated benchmark.
We parametrize the approximation by the {\em grid size}  $\Delta$ of the discrete distribution, defined as $\Delta=\sup_{v\in V}\inf_{v'\in V}\{v-v'~|~v'< v\}$.
For a distribution with a single support point, define $\Delta=\infty$.
We prove the lemma in Appendix~\label{app:discreterevenue}

\begin{lemma}\label{lem:discreterevenue}
	For any discrete regular distribution with grid size $\Delta$, $$\max_{p\in V}p(1-F( p/(1-\delta)))\geq (1-\delta)\sophprob\max_{p\in V}p(1-F_+^S(p))-\Delta.$$
\end{lemma}


\begin{proof}[Theorem~\ref{thm:revenue}]
	Let $p$ maximize $\tfrac{\delta}{1-\delta}\hat p(1-F(\hat p/(1-\delta)))$ over $\hat p\in V$.
	We will exhibit a family of deviation strategies, parametrized by small error term $\epsilon>0$.
	For each $\epsilon$, we show that the seller's revenue is at least $\tfrac{\delta}{1-\delta} p(1-F(( p+\epsilon\delta)/(1-\delta)))$ by analyzing the sophisticated buyers' threshold response.
	Since the seller chooses their equilibrium strategy to maximize their expected discounted revenue, their equilibrium revenue is at least this bound for any $\epsilon$.
	This implies the bound stated in the theorem.
	
	The deviation strategy is as follows.
	In period $1$, the seller sets $p_1=p+\epsilon$, with $\epsilon$ selected such that there is at least one point in $V$ above $ p+\epsilon$.
	(Such a choice of $\epsilon$ must exist if $\tfrac{\delta}{1-\delta} p(1-F(p/(1-\delta)))>0$, as this implies $1-F( p/(1-\delta))>0$.)
	In period $2$ following an accept, set $p_2^\acc$ to be the infimum support point above $p_1$, given by $p_2^\acc=\inf\{v\in V~|~v\geq p+\epsilon\}$, then play as prescribed by the equilibrium.
	On reject, set $p_2^\rej=p$, for $\epsilon>0$, then play as prescribed.
	
	Analysis of this deviation strategy will proceed in two steps.
	We will first show that the seller's revenue conditioned on either an accept in period $1$ or a reject in period $1$ followed by an accept in period $2$ is at least $\tfrac{\delta}{1-\delta} p$.
	We will then show that the probability that neither of these events occurs is at most $1-F((p+\delta \epsilon)/(1-\delta))$.
	Together, these steps imply that the expected revenue from the deviation strategy is at least $\tfrac{\delta}{1-\delta} p(1-F(( p+\epsilon\delta)/(1-\delta)))$.
	
	We first analyze the revenue conditioned on the event that the buyer accepts in period $1$.
	The seller gets $p_1$ from period $1$, followed by their expected revenue from the continuation game starting in period $2$.
	In the continuation game, the seller believes that if the buyer is naive, their value must be at least $p_1=p+\epsilon$.
	Since there is at least one value in $V$ above $p_1$, \cb\ implies that the seller's expected revenue starting from period $2$ is at least $\tfrac{1}{1-\delta} (p+\epsilon)$.
	Hence, the seller's revenue conditioned on a first period accept is at least $p+\epsilon+\tfrac{\delta}{1-\delta} (p+\epsilon)>\tfrac{\delta}{1-\delta}p$.
	
	Next, consider the event where the buyer rejects in period $1$, but accepts in period $2$.
	In this case, since $p_2^\rej=p$ is in the support $V$, the seller believes that the seller could be naive after period $2$, and if so, that their value is at least $p$.
	Hence by \cb, the seller's expected revenue starting from period $3$ is at least $\tfrac{1}{1-\delta}p$, and their discounted revenue starting from period $1$ in this event is $\delta p$ from period $2$ plus the continuation revenue starting in period $3$, yielding total discounted revenue at least $\delta p+\tfrac{\delta^2}{1-\delta}p=\tfrac{\delta}{1-\delta}p$.
	
	Now we show that the probability the buyer rejects in both periods $1$ and $2$ is at most $\prob[v\leq (p+\delta\epsilon)/(1-\delta)]$.
	We analyze the events that the buyer is sophisticated ($\mathcal E_S$) and naive ($\mathcal E_N)$ separately.
	Conditioned on $\mathcal E_S$, we prove that the probability of rejection in both rounds is at most $F^S( (p+\epsilon\delta)/(1-\delta))$.
	Conditioned on $\mathcal E_N$, it is straightforward to prove the even stronger bound on the two-round rejection probability of $F^N( p/(1-\delta))$.
	The law of total probability then implies an overall upper bound on probability of rejection of $F( (p+\epsilon\delta)/(1-\delta))$, as desired.
	
	We now analyze $\mathcal E_S$ and $\mathcal E_N$. First condition on the buyer being sophisticated.
	Let $t_1$ denote the threshold response to $p_1$, and $t_2^\rej$ the threshold response to $p_2^\rej$ in period $2$.
	A sophisticated buyer rejects in both periods if $v\leq \min(t_1,t_2^\rej)$.
	Note that $F^S(\min(t_1,t_2^\rej))\leq F^S(t_2^\rej)$. 
	We upper bound this latter probability by upper bounding $t_2^\rej$.
	To do so, we compare the utility of the threshold buyer with value $t_2^\rej$ for accepting with that of the same buyer for rejecting.
	By \ca, all prices following a period $2$ accept will be at most $p+\epsilon$.
	Hence, the threshold buyer's utility for accepting in period $2$, starting from that period, is at least $t_2^\rej-p+\tfrac{\delta}{1-\delta}(t_2^\rej-p-\epsilon)$.
	Meanwhile, since prices are nonnegative, a crude upper bound on the utility of the threshold type for rejecting is $\tfrac{\delta}{1-\delta}t_2^\rej$.
	Since the threshold type is indifferent between accepting and rejecting, we have
	\begin{equation*}
		t_2^\rej-p+\tfrac{\delta}{1-\delta}(t_2^\rej-p-\epsilon)\leq \tfrac{\delta}{1-\delta}t_2^\rej,
	\end{equation*}
	which can be rearranged to obtain $t_2^\rej\leq (p+\epsilon\delta)/(1-\delta)$.
	Hence, $F^S(\min(t_1,t_2^\rej))\leq F^S(t_2^\rej)\leq F^S( (p+\epsilon\delta)/(1-\delta))$.
	This is the desired bound.
	
	Finally, condition on the buyer being naive.
	The probability a naive buyer rejects both $p+\epsilon$ and $p$ is equal to the probability that they reject the lower price of $p$, which is $\prob[v<p~|~\mathcal E_N]\leq F^N( p)\leq F^N( p/(1-\delta))$.
	This is the desired bound, and implies the theorem.
\end{proof}

Combining the three bounds of this section, we have the following:
\begin{corollary}
	For any threshold PBE for continuous beliefs with positive probability of naive buyers additionally satisfying \ca\ and \cb, the seller's discounted revenue is at least
	\begin{equation*}
		\max(\max_{p\in V}\tfrac{\delta}{1-\delta}p(1-F(\tfrac{p}{1-\delta})),\tfrac{\ell}{1-\delta},(1-\sophprob)\rev(F^N)).
	\end{equation*}
\end{corollary}
For $\delta$ bounded away from $1$ and continuous distributions over $[0,H]$, this implies a multiplicative constant-approximation to the commitment benchmark.
For discretely distributed values, we obtain a constant-approximation with additive error.

\subsection{Discussion: Full-Sophistication Refinement}
\label{sec:infrefine}

We conclude with a final observation about the equilibrium in Section~\ref{sec:infeq} that has implications for the all-sophisticated setting.
Recall that the probability of naive buyers was fixed at some sufficiently small, but otherwise unspecified $\epsilon>0$.
In the description and analysis of the equilibrium, the precise value of $\epsilon$ turned out to be unimportant: there exists a neighborhood about $0$ such that for all $\epsilon$ in this neighborhood, the strategies described in Section~\ref{sec:infeq} form a PBE.
This allows us to say something stronger, though: even for $\epsilon$ {\em equal to} $0$, the strategies we describe are a PBE.\footnote{This requires a small change to handle cases that are now off-path with the complete elimination of naive buyers. See Appendix~\ref{app:zeronaive} for the details of this zero-naive version of the equilibrium.}
In other words Section~\ref{sec:infeq} demonstrates a natural PBE of for the infinite-horizon setting where the seller obtains nontrivial revenue.

This seems to contradict to the results of \citet{ilpt17}, who argue that the zero-learning equilibrium is focal.\footnote{The subsequent discussion glosses over the fact that \citet{ilpt17} only prove their result for continuous distributions.}
The argument of \citet{ilpt17} is based on three refinements: (1) threshold equilibrium (2) Markovian equilibrium, and (3) a requirement that prices be ``undominated,'' as we shall define.
The equilibrium of Section~\ref{sec:infeq} satisfies requirements (1) and (2).
The undominated prices requirement states that the seller never offers prices outside the support of the sophisticated buyers, on the ad hoc justification that such prices would be dominated in the single-round game.
This property notably {\em does not} hold if we consider the strategies outlined in Section~\ref{sec:infeq} and take the naive probability $\epsilon$ to be $0$.
In particular, after observing a period-$1$ accept, the seller offers a price of $10$, despite knowing the sophisticated buyer's value is $20$.
However, as outlined in the equilibrium description, the sophisticated buyer rejects any $p_2^\acc$ above $10$, so offering $10$ is sequentially rational for the seller (and we argue in the appendix that this rejection behavior is also sequentially rational for the buyer).
This low accept price then serves to incentivize high-value buyers to identify themselves by accepting in period $1$, leading to higher revenue.

We have already observed that the zero-learning equilibrium is not robust to the introduction of naive buyers.
It is arguably undesirable for an equilibrium refinement to select a non-robust equilibrium, and it has been noted by previous work on the problem that this equilibrium is counterintuitive.
The discussion above suggests a remedy to this non-robustness and counterintuitiveness.
A seller who suspects the buyer of naivete may be justified in offering prices that would be dominated otherwise.
Hence, it might be reasonable to discard the ``undominated prices'' refinement, and instead select heterogeneous-population equilibria which survive the removal of all naive buyers.
We leave further study of this proposed refinement (e.g.\ connecting it to existing refinements from the signaling literature) to future work.

\bibliographystyle{ACM-Reference-Format}
\bibliography{refs}

\appendix
\section{Linear Demand Case: Full Description and Analysis}
\label{app:uniform}
This appendix gives more details on the linear-demand equilibrium, where $v\sim U[0,1]$. 
Algorithms~\ref{alg:seller} and \ref{alg:buyer} contain the full descriptions of the strategies for the seller and sophisticated buyer, respectively.
As discussed in Section~\ref{sec:prelims}, beliefs can all be obtained via Bayesian updating, using formulas (\ref{eq:frej}) and (\ref{eq:facc}), and there are no off-path continuations because of the presence of naive buyers. The strategies described are valid for all $\mu\in[0,1)$. For $\mu=1$, we simply use the equilibrium computed in \citet{dps19}, which is the limit point of our computed equilibrium in some reasonable sense, but requires the specification of off-path beliefs due to the lack of naive buyers.

It is relatively straightforward, albeit involved, to derive equilibrium via backward induction.
We describe the major steps and give the important intermediate formulas, but any remaining steps follow from elementary computation.

We begin with the seller's pricing problem in period $2$. 
Assume in period $1$ the seller posted a price of $p_1$, and the sophisticated buyers responded according to a threshold $t$.
The seller's second-period revenue from a price $p$ is given by the formulas $R^D(p;p_1,t)=p(1-F^D(p;p_1,t))$, for $D\in \{\rej,\acc\}$, where we make the dependence of $F^D$ on $p_1$ and $t$ explicit.
Plugging in the formulas (\ref{eq:frej}) and (\ref{eq:facc}) and optimizing with respect to price yields the following second-round optimal prices, assuming $p_1\leq t$ as will be the case in our resulting equilibrium:
\begin{align*}
p_2^\rej(p_1,t,\mu)=\begin{cases}
\frac{p_1 (1 - \mu) + t \mu}{2} & t \leq p_1 \left( 1 + \frac{1}{\sqrt{\mu}} \right)\\
\frac{t}{2} & t > p_1 \left( 1 + \frac{1}{\sqrt{\mu}} \right)
\end{cases} && p_2^\acc(p_1,t,\mu)=\begin{cases}
p_1 & \frac{1 - t \mu}{2 (1 - \mu)} \leq p_1\\
\frac{1 - t \mu}{2 (1 - \mu)} & p_1 \leq \frac{1 - t \mu}{2 (1 - \mu)} \leq t\\
t & \frac{1}{2} \leq t \leq \frac{1 - t \mu}{2 (1 - \mu)}\\
\frac{1}{2} & t \leq \frac{1}{2}
\end{cases},
\end{align*}
where the seller has two distinct options for $p_2^\rej$ at the boundary case $t=p_1(1+1/\sqrt\mu)$.

We can now compute the buyer's continuation utilities after round $1$, as a function of their threshold strategy $t$.
Given $p_1$ and $\mu$, a choice of $t$ implies a unique choice of $p_2^\acc$ and $p_2^\rej$ accept at the boundary case for $p_2^\rej$ boundary case, in which case the seller has two choices.
Assuming the second-round prices are single-valued, the utility of the marginal type $t$ from accepting will be $t-p_1+(t-p_2^\acc(p_1,t,\mu))^+$, and from rejecting, $t-p_2^\rej(p_1,t,\mu)$, where the latter formula uses the fact that $p_2^\rej$ computed above lies below $t$. 
These continuation utilities can be computed as:
\begin{align*}
U^\rej(p_1,t,\mu)=\begin{cases}
t - \frac{p_1 (1 - \mu) + t \mu}{2} & t \leq p_1 \left( 1 + \frac{1}{\sqrt{\mu}} \right)\\
\frac{t}{2} & t > p_1 \left( 1 + \frac{1}{\sqrt{\mu}} \right)
\end{cases} && U^\acc(p_1,t,\mu)=\begin{cases}
2(t - p_1) & \frac{1 - t \mu}{2 (1 - \mu)} \leq p_1\\
2t - p_1 - \frac{1 - t \mu}{2 (1 - \mu)} & p_1 \leq \frac{1 - t \mu}{2 (1 - \mu)} \leq t\\
t - p_1 & t < \frac{1}{\mu + 2 (1 - \mu)}
\end{cases}.
\end{align*}

Because $p_2^\rej$ is discontinuous, for a large range of $\mu$, $U^\rej$ and $U^\acc$ do not cross.
When they do, however, the $t$ where they cross will solve the buyer indifference equation.
When they do not, the seller must play a mixed strategy in round $2$ to induce indifference.
More specifically, for $p_1$ and $\mu$ where $U^\rej$ and $U^\acc$ do not cross, choose $t=p_1(1+1/\sqrt\mu)$ --- that is, we select the $t$ that makes the seller indifferent between two choices of $p_2^\rej$.
By choosing the proper mixture probability between these two prices, we can induce a value of $t-\mathbb E[p_2^\rej(p_1,t,\mu)]$ equal to $t-p_1+(t-p_2^\acc(p_1,t,\mu))^+$.
These mixture probabilities can be found in the round $2$ strategy given in Algorithm~\ref{alg:seller}.
Following the steps described, we obtain the following threshold as a function of $\mu$ and $p_1$:
\begin{equation*}
    t(p_1,\mu)=\begin{cases}
    p_1 \left( 1 + \frac{1}{\sqrt{\mu}} \right) & p_1 \leq \frac{\sqrt{\mu}}{2 + \sqrt{\mu} - \mu^2}\\
    \frac{1 + p_1 (1 - \mu^2)}{2 - \mu^2} & \frac{\sqrt{\mu}}{2 + \sqrt{\mu} - \mu^2} < p_1 \leq \frac{2 + \mu}{4 + \mu - \mu^2}\\
    p_1 \left( \frac{3 + \mu}{2 + \mu} \right) & \frac{2 + \mu}{4 + \mu - \mu^2} < p_1 \leq \frac{2 + \mu}{3 + \mu}\\
    1 & \frac{2 + \mu}{3 + \mu} < p_1
    \end{cases}.
\end{equation*}

Finally, we must compute the seller's optimal choice of $p_1$.
Given $\mu$ and a choice of $p_1$ and taking $t$, $p_2^\rej$, and $p_2^\acc$ according to the formulas derived above, we have the following expected revenue:
\begin{align*}
    \Rev(p_1,\mu)&=\mu\left(p_1(1-p_1)+\mathbb E[p_2^\rej(p_1-p_2^\rej)^+]+p_2^\acc(1-p_2^\acc)\right)\\
    &\quad\quad\quad+(1-\mu)\left(p_1(1-t)+\mathbb E[p_2^\rej(t-p_2^\rej)]+p_2^\acc(1-p_2^\acc)\right).
\end{align*}
Working out this expression yields the following:
\begin{equation*}
\Rev(p_1,\mu)=\begin{cases}
\frac{1}{4} + p_1 + \frac{1}{4} p_1^2 \left(-3 - 2 \sqrt{\mu} + \mu\right) & p_1 \leq \frac{1}{2 \left( 1 + \frac{1}{\sqrt{\mu}} \right)}\\
\frac{p_1 \left(4 \left(\sqrt{\mu} + 2 \mu\right) + p_1 \left( -4 - 8 \sqrt{\mu} - 7 \mu - 2 \mu^{3/2} + \mu^2 \right)\right)}{4 \mu} & \frac{1}{2 \left( 1 + \frac{1}{\sqrt{\mu}} \right)} < p_1 \leq \frac{\sqrt{\mu}}{\left(1 + \sqrt{\mu}\right)(2 - \mu)}\\
\frac{-1 + 2 p_1 \left( -2 + \sqrt{\mu} + 3 \mu \right) + p_1^2 \left( 3 + 2 \sqrt{\mu} - 5 \mu - 4 \mu^{3/2} \right)}{4(-1 + \mu)} & \frac{\sqrt{\mu}}{\left(1 + \sqrt{\mu}\right)(2 - \mu)} < p_1 \leq \frac{\sqrt{\mu}}{2 + \sqrt{\mu} - \mu^2}\\
-\frac{-4 + 2 \mu^2 + \mu^3 - 2 p_1 (8 - 4 \mu - 10 \mu^2 + 3 \mu^2 + 3 \mu^4) + p_1^2 (12 - 4 \mu - 14 \mu^2 + \mu^3 + 4 \mu^4 + \mu^5)}{4(-2 + \mu^2)^2} & \frac{\sqrt{\mu}}{2 + \sqrt{\mu} - \mu^2} < p_1 \leq \frac{2 + \mu}{4 + \mu - \mu^2}\\
p_1\left(2 - \frac{p_1 (7 + 10 \mu + 3 \mu^2)}{(2 + \mu)^2}\right) & \frac{2 + \mu}{4 + \mu - \mu^2} < p_1 \leq \frac{2 + \mu}{3 + \mu}\\
2 (-1 + p_1 ) p_1 (-1 + \mu) + \frac{1}{4} (p_1 + \mu - p_1 \mu)^2 & \frac{2 + \mu}{3 + \mu} < p_1
\end{cases}.
\end{equation*}
The resulting revenue curve generally has two local maxima $p^*_-$ and $p^*_+$, which vary in quality depending on $\mu$.
For low $\mu$, $p^*_+$ dominates, and for high $\mu$, $p^*_-$ dominates.
The point where the two maxima are equal is when
\begin{equation*}
\tfrac{-7-2\mu+5\mu^2+2\mu^3}{-12-8\mu+6\mu^2+5\mu^3+\mu^4}=\tfrac{(1+2 \mu^{1/2})^2}{4+8 \mu^{1/2}+7 \mu+2 \mu^{3/2}- \mu^2},
\end{equation*}
which occurs at $\mu\approx .630209$.

\SetKw{roundone}{\underline{Round 1 Strategy}} 
\SetKw{roundtwo}{\underline{Round 2 Strategy}} 
\begin{algorithm}
\DontPrintSemicolon
\roundone\;
\KwIn{Sophistication level $\mu$.}
\KwOut{First-round price $p_1$.}\;
Let $\hat \mu=\tfrac{1}{3}\left(\sqrt[3]{\tfrac{1}{2}\left(43-3\sqrt{177}\right)}+\sqrt[3]{\tfrac{1}{2}\left(43+3\sqrt{177}\right)}\right)\approx 0.205569$.\;
Let $\bar \mu\approx .630209$ be the unique solution to $\tfrac{-7-2\overline\mu+5\overline\mu^2+2\overline\mu^3}{-12-8\overline\mu+6\overline\mu^2+5\overline\mu^3+\overline\mu^4}=\tfrac{(1+2\overline \mu^{1/2})^2}{4+8\overline \mu^{1/2}+7\overline \mu+2\overline \mu^{3/2}-\overline \mu^2}$ in $[0,1]$.\;
\lIf{$\mu<\hat\mu$}
{
$p_1=\tfrac{(2+\mu)^2}{7+10\mu+3\mu^2}$
}
\lElseIf{$\hat \mu\leq \mu < \tfrac{1}{2}$}
{
$p_1=\tfrac{2+\mu}{4+\mu-\mu^2}$
}
\lElseIf{$1/2\leq \mu < \bar\mu$}
{
$p_1=\tfrac{-8-4\mu+6\mu^2+3\mu^3}{-12-8\mu+6\mu^2+5\mu^3+\mu^4}$
}
\lElse{
$p_1=\tfrac{2\mu^{1/2}+4\mu}{4+8\mu^{1/2}+7\mu+2\mu^{3/2}-\mu^2}.$\;
}

\roundtwo\;
\KwIn{Sophistication level $\mu$, first round price $p_1$, first round {\em Accept/Reject} decision $D$.}
\KwOut{Second round (randomized) price $p_2$.}\;
\uIf{$D=\text{Accept}$}{
\lIf{$p_1< \tfrac{\sqrt{\mu}}{2(1+\sqrt{\mu})}$}{$p_2=\tfrac{1}{2}$}
\lElseIf{$\tfrac{\sqrt{\mu}}{2(1+\sqrt{\mu})}\leq p_1<\tfrac{\sqrt{\mu}}{(2-\mu)(1+\sqrt{\mu})}$}{$p_2=p_1\tfrac{1+\sqrt\mu}{\sqrt\mu}$}
\lElseIf{$\tfrac{\sqrt{\mu}}{(2-\mu)(1+\sqrt{\mu})}\leq p_1<\tfrac{\sqrt\mu}{2+\sqrt\mu-\mu^2}$}{$p_2=\tfrac{\sqrt\mu-p_1\mu(1+\sqrt\mu)}{2(1-\mu)\sqrt\mu}$}
\lElseIf{$\tfrac{\sqrt\mu}{2+\sqrt\mu-\mu^2}\leq p_1<\tfrac{2+\mu}{4+\mu-\mu^2}$}{$p_2=\tfrac{2+\mu-p_1\mu(1+\mu)}{2(2-\mu^2)}$}
\lElse{$p_2=p_1$}
}
\uElseIf{$D=\text{Reject}$}{
\uIf{$p_1< \tfrac{\sqrt \mu}{(2-\mu)(1+\sqrt\mu)}$}{
with probability $\tfrac{1}{1+\sqrt\mu}$, offer $p_2=\tfrac{1}{2}p_1(1+\sqrt\mu)$\;
with probability $\tfrac{\sqrt \mu}{1+\sqrt \mu}$, offer $p_2=p_1\tfrac{1+\sqrt\mu}{2\sqrt\mu}$
}
\uElseIf{$\tfrac{\sqrt \mu}{(2-\mu)(1+\sqrt\mu)}\leq p_1<\tfrac{\sqrt\mu}{2+\sqrt\mu-\mu^2}$}
{
with probability $\tfrac{3p_1-\sqrt\mu+p_1\sqrt\mu-2p_1\mu}{p_1(1-\mu)^2}$, offer $p_2=\tfrac{1}{2}p_1(1+\sqrt\mu)$\;
with probability $1-\tfrac{3p_1-\sqrt\mu+p_1\sqrt\mu-2p_1\mu}{p_1(1-\mu)^2}$, offer $p_2=p_1\tfrac{1+\sqrt\mu}{2\sqrt\mu}$
}
\lElseIf{$\tfrac{\sqrt\mu}{2+\sqrt\mu-\mu^2}\leq p_1<\tfrac{2+\mu}{4+\mu-\mu^2}$}{$p_2=\tfrac{\mu(2+\mu)+p_1(4-4\mu-3\mu^2+\mu^3)}{4(2-\mu^2)}$}
\lElseIf{$\tfrac{2+\mu}{4+\mu-\mu^2}\leq p_1<\tfrac{2+\mu}{3+\mu}$}{$p_2=p_1\tfrac{1+\mu}{2+\mu}$}
\lElse{$p_2=\tfrac{\mu+p_1(1-\mu)}{2}.$}
}

\caption{{\sc Seller Strategy}}
\label{alg:seller}
\end{algorithm}

\SetKw{roundone}{\underline{Round 1 Strategy}} 
\SetKw{roundtwo}{\underline{Round 2 Strategy}} 
\begin{algorithm}
\DontPrintSemicolon
\roundone\;
\KwIn{Sophistication level $\mu$, first-round price $p_1$.}
\KwOut{First-round threshold $t_1$.}\;

\lIf{$p_1<\tfrac{\sqrt\mu}{2+\sqrt\mu-\mu^2}$}{
$t_1=p_1\tfrac{1+\sqrt\mu}{\sqrt\mu}$
}
\lElseIf{$\tfrac{\sqrt\mu}{2+\sqrt\mu-\mu^2}\leq p_1<\tfrac{2+\mu}{4+\mu-\mu^2}$}
{
$t_1=\tfrac{1+p_1(1-\mu^2)}{2-\mu^2}$
}
\lElseIf{$\tfrac{2+\mu}{4+\mu-\mu^2}\leq p_1<\tfrac{2+\mu}{3+\mu}$}{
$t_1=p_1\tfrac{3+\mu}{2+\mu}$
}
\lElse{$t_1=1$.}\;

\roundtwo\;
\KwIn{Sophistication level $\mu$, prices $p_1,p_2$, first round {\em Accept/Reject} decision $D$.}
\KwOut{Second round threshold $t_2$.}\;
Set $t_2=p_2$.

\caption{{\sc Buyer Threshold Strategy}}
\label{alg:buyer}
\end{algorithm}
\section{Proof of Lemma~\ref{lem:p1implement}}
\label{app:exists}

Fix $p_1$.
Define the seller's accept/reject price correspondences as a function of the threshold $t$ as follows:
\begin{itemize}
    \item $\mathcal P_2^\rej(t)=\arg\max_p R^\rej(p\,;t)$
    \item $\mathcal P_2^\acc(t)=\arg\max_p R^\acc(p\,;t)$,
\end{itemize}
where we make the dependence of the revenue curves on $t$ explicit, and dependence on $p_1$ implicit.
Note that these correspondences are well-defined even when $t$ exceeds $1$.
From the definition of the conditional revenue curves, both correspondences are upper hemicontinuous with respect to $t$.
Further note that for all $t$, and all $p\in\mathcal P_2^\rej(t)\cup\mathcal P_2^\acc(t)$, $p\leq 1$.

We may further define the buyer's {\em implementable utility correspondences} in the following way:
\begin{itemize}
    \item $\mathcal U^\rej(t)=\{(t-p_2^\rej)^+\,|\,p_2^\rej\in P_2^\rej(t)\}$
    \item $\mathcal U^\acc(t)=\{t-p_1+(t-p_2^\acc)^+\,|\,p_2^\acc\in P_2^\acc(t)\}$
\end{itemize}
Observe that $\mathcal U^\rej$ and $\mathcal U^\acc$ are also upper hemicontinuous. Moreover, observe that $\mathcal U^\rej(0)=0$, $\mathcal U^\rej(0)=-p_1$, and that there exists some $\overline t$ sufficiently large (possibly greater than $1$) such that every $u^\rej\in\mathcal U^\rej(\overline t)$ and $u^\acc\in\mathcal U^\acc(\overline t)$, $u^\rej\leq u^\acc$. Denote this latter relationship by $\mathcal U^\rej(\overline t)\leq \mathcal U^\acc(\overline t)$.

Consider $t^*=\inf\{t\,|\,\mathcal U^\rej(t)\leq \mathcal U^\acc(t)\}$. By the upper hemicontinuity of $\mathcal U^\rej$ and $\mathcal U^\acc$, one or both of the following must hold:
\begin{itemize}
    \item there exists $u^\acc\in\mathcal U^\acc(t^*)$ and ${\underline u}^\rej,{\overline u}^\rej\in \mathcal U^\rej(t^*)$ such that ${\underline u}^\rej\leq u^\acc\leq{\overline u}^\rej$.
    \item there exists $u^\rej\in\mathcal U^\rej(t^*)$ and ${\underline u}^\acc,{\overline u}^\acc\in \mathcal U^\acc(t^*)$ such that ${\underline u}^\acc\leq u^\rej\leq{\overline u}^\acc$.
\end{itemize}
We consider the former case: the latter follows from identical reasoning.
Define $p_2^\acc$ to be the price certifying $u^\acc\in\mathcal U^\acc(t^*)$, and ${\underline p}_2^\rej,{\overline p}_2^\rej$ the prices certifying ${\underline u}^\rej,{\overline u}^\rej\in \mathcal U^\rej(t^*)$, respectively.
Consider the continuation 
$(p_1,p_2^\rej,p_2^\acc,t^*)$\footnote{A threshold greater than $1$ has the interpretation that all sophisticated types reject.}, where $p_2^\rej$ is ${\overline p}_2^\rej$ with probability $\alpha$ and ${\underline p}_2^\rej$ with probability $1-\alpha$, for $\alpha=(u^\acc-{\underline u}^\rej)/({\overline u}^\rej-{\underline u}^\rej)$. This continuation then satisfies conditions \ref{item:ropt}, \ref{item:aopt}, and \ref{item:threshold} of Definition~\ref{def:continuation} by construction.

\begin{proof}[Proof of Theorem~\ref{thm:exists}]
For any $p_1$, Lemma~\ref{lem:p1implement} gives an implementation where $t$ satisfies the buyer indifference equation $t-p_1+\mathbb E[(t-p_2^\acc)^+]=t-\mathbb E[p_2^\rej]$. Moreover, for any $v>t$, the lefthand side is weakly greater than the righthand side, i.e.\ $v-p_1+\mathbb E[(v-p_2^\acc)^+]=v-\mathbb E[p_2^\rej]$, and for any $v\leq t$, the lefthand side is weakly less than the righthand side.
Hence, all sophisticated types are incentivised to abide by the threshold of $t$.
The final observation is that by the construction of the continuation in Lemma~\ref{lem:p1implement}, the seller's revenue from the continuation varies continuously with $p_1$. Consequently, there exists a revenue-optimal $p_1$ for the seller to pick.
\end{proof}
\section{Proof of Lemma~\ref{lem:ordering}}
\label{app:ordering}

We prove the main lemma using the lemmas below.

\begin{lemma}
\label{lem:orderingp2rp*}
    Let $F$ have monopoly reserve $p^*$. Then for any continuation $(p_1, p_2^\mathcal{R}, p_2^\mathcal{A},t )$, $p_2^\mathcal{R} \leq p^*$ for every realization of $p_2^\mathcal{R}$ and $p_2^\mathcal{A}$.
\end{lemma}
\begin{proof}
    Assume $p_1\leq t$. A symmetric argument will hold for the case where $p_1>t$. We may further assume that $p^*< t$, as the claim is trivial otherwise. We will pick a arbitrary price $p>p^*$, and show that $R^\rej(p)<R^\rej(p^*)$. The revenue of any price conditioned on reject is a convex combination of $R_{\leq t}$ and $R_{\leq p_1}$. We will first show that $R_{\leq t}(p^*)>R_{\leq t}(p)$, and $R_{\leq p_1}(p^*)\geq R_{\leq p_1}(p)$, with equality only if $p^*\geq p_1$.

    We will first argue for $R_{\leq t}$.
    If $p\geq t$, the claim holds trivially. Otherwise, both $p$ and $p^*$ are less than $t$. We can write:
    \begin{equation*}
        \frac{R_{\leq t}(p^*)}{R_{\leq t}(p)}=\frac{p^*(1-F(p^*)/F(t))}{p(1-F(p)/F(t))}=\frac{p^*(F(t)-F(p^*))}{p(F(t)-F(p))}> \frac{p^*(1-F(p^*))}{p(1-F(p))}>1.
    \end{equation*}
    Where the first inequality follows from the fact that $F(p)> F(p^*)$ and the second from the fact that $R(p^*)> R(p)$. If $p^*<p_1$, an identical argument holds. If $p^*\geq p_1$, then $p\geq p_1$, and $R_{\leq p_1}(p^*)= R_{\leq p_1}(p)=0$.

    We now show that $R^\rej(p)<R^\rej(p^*)$. If $p\leq p_1$, then we may write the revenue difference as:
    \begin{equation*}
        \mu_\rej (R_{\leq t}(p^*)-R_{\leq t}(p))+(1-\mu_\rej) (R_{\leq p_1}(p^*)-R_{\leq p_1}(p))> 0.
    \end{equation*}
    If $p^*\leq p_1 < p$, then the revenue difference is:
        \begin{equation*}
        \mu_\rej (R_{\leq t}(p^*)-R_{\leq t}(p))+(1-\mu_\rej) R_{\leq p_1}(p^*)> 0.
    \end{equation*}
    Finally, if $p^*>p_1$, the revenue difference is
        \begin{equation*}
        \mu_\rej (R_{\leq t}(p^*)-R_{\leq t}(p))> 0.
    \end{equation*}
    The lemma therefore holds in each case.
\end{proof}

\begin{lemma}
\label{lem:orderingp*p2a}
    Let $F$ have monopoly reserve $p^*$. Then for any continuation $(p_1, p_2^\mathcal{R}, p_2^\mathcal{A},t )$, $p^* \leq p_2^\mathcal{A}$ for every realization of $p_2^\mathcal{R}$ and $p_2^\mathcal{A}$.
\end{lemma}
\begin{proof}
    Assume $p_1\leq t$. A symmetric argument will hold for the case where $p_1>t$. We may further assume that $p^*>p_1$, as the claim is trivial otherwise. We will pick an arbitrary price $p<p^*$, and show that $R^\acc(p)<R^\acc(p^*)$. The revenue of any price conditioned on accept is a convex combination of $R_{\geq t}$ and $R_{\geq p_1}$. We will first show that $R_{\geq p_1}(p^*)>R_{\geq p_1}(p)$, and $R_{\geq t}(p^*)\geq R_{\geq t}(p)$, with equality only if $p^*\leq t$.

    We will first argue for $R_{\geq p_1}$.
    We can write:
    \begin{equation*}
        \frac{R_{\geq p_1}(p^*)}{R_{\geq p_1}(p)}=\frac{p^*\left(\frac{1-F(p^*)}{1-F(t)}\right)}{p\left(\frac{1-F(p)}{1-F(t)}\right)}= \frac{p^*(1-F(p^*))}{p(1-F(p))}>1.
    \end{equation*}
    Where the inequality follows from the fact that $R(p^*)> R(p)$. If $p^*>t$, an identical argument holds for $R_{\geq t}$. If $p^*\leq t$, then $p\leq t$, and $R_{\geq t}(p^*)= R_{\geq t}(p)=0$.

    We now show that $R^\acc(p)<R^\acc(p^*)$. If $p\geq p_1$, then we may write the revenue difference as:
    \begin{equation*}
        \mu_\acc (R_{\geq t}(p^*)-R_{\geq t}(p))+(1-\mu_\acc) (R_{\geq p_1}(p^*)-R_{\geq p_1}(p))> 0.
    \end{equation*}
    If $p^*\geq p_1 > p$, then the revenue difference is:
        \begin{equation*}
        \mu_\acc (R_{\geq t}(p^*)-R_{\geq t}(p))+(1-\mu_\acc) R_{\geq p_1}(p^*)> 0.
    \end{equation*}
    Finally, if $p^*<p_1$, the revenue difference is
        \begin{equation*}
        \mu_\acc (R_{\geq t}(p^*)-R_{\geq t}(p))> 0.
    \end{equation*}
    The lemma therefore holds in each case.
\end{proof}

\begin{lemma}
\label{lem:orderingp1t}
For any continuation $(p_1, p_2^\mathcal{R}, p_2^\mathcal{A},t )$, $p_1 \leq t$ for every realization of $p_2^\mathcal{R}$ and $p_2^\mathcal{A}$.
\end{lemma}
\begin{proof}
The following threshold equation must be satisfied:
\begin{align*}
    & t - p_1 + \max \{0, t-p_2^\mathcal{A}\} = \max \{0, t - p_2^\mathcal{R}\}\\
    \iff\ & p_1 = t - \max \{0, t - p_2^\mathcal{R}\} + \max \{0, t-p_2^\mathcal{A}\}
\end{align*}
To show $t \geq p_1$, it is sufficient to show that $\max \{0, t-p_2^A\} \leq \max \{0, t - p_2^R\}$. We can show this using the fact that $p_2^\mathcal{R} \leq p_2^\mathcal{A}$ which is implied by Lemma~\ref{lem:orderingp2rp*} and Lemma~\ref{lem:orderingp*p2a}.
\end{proof}

\begin{lemma}
For any continuation $(p_1, p_2^\mathcal{R}, p_2^\mathcal{A},t )$, $p_2^\mathcal{R} \leq t$ for every realization of $p_2^\mathcal{R}$ and $p_2^\mathcal{A}$.
\end{lemma}
\begin{proof}
    By Lemma~\ref{lem:orderingp1t}, we know that $p_2^\acc \leq t$; otherwise, $R^\acc(p_2^\acc) = 0$ which is a contradiction.
\end{proof}

\begin{lemma}
For any continuation $(p_1, p_2^\mathcal{R}, p_2^\mathcal{A},t )$, $p_2^\mathcal{A} \geq p_1$ for every realization of $p_2^\mathcal{R}$ which is impossible because we can offer a price that obtains positive revenue.
\end{lemma}
\begin{proof}
    By Lemma~\ref{lem:orderingp1t}, we know that $p_2^\rej \geq p_1$; otherwise, $R^\rej(p_2^\rej) = 0$ which is a contradiction.
\end{proof}

\section{Proof of Lemma~\ref{lem:plph}}
\label{app:mixture}

To prove the lemma, we will first characterize the two pieces of the reject revenue curve $R^\rej$.
\label{sec:mixed}
\begin{lemma}\label{lem:trunccon}
	Let $F$ satisfy Assumption~\ref{ass:regular}. Then for any $x> 0$, the truncated distribution $F_{\leq x}$ satisfies Assumption~\ref{ass:regular} as well.
\end{lemma}
\begin{proof}
	Assumption~\ref{ass:regular} is that the revenue curve $R(p)=p(1-F(p))$ is strictly concave. For $p\in [0,x]$, we may write the conditional revenue curve as
	\begin{equation*}
		R_{\leq x}(p)=p(1-F_{\leq x}(p))
		=p\left(1-\tfrac{F(p)}{F(x)}\right)
		=\tfrac{p(1-F(p))}{F(x)}+\tfrac{p(F(x)-1)}{F(x)}
		=\tfrac{R(p)}{F(x)}+\tfrac{p(F(x)-1)}{F(x)}
	\end{equation*}
	The first term on the righthand side is strictly concave, since $R(p)$ is strictly concave. The second term is linear in $p$. Hence, the sum is strictly concave.
\end{proof}

\begin{lemma}\label{lem:pwconcave}
	Let $F$ satisfy Assumption~\ref{ass:regular}. Then given $p_1\leq t$, the second-round revenue curve conditioned on reject $R^\rej=p(1-F^\rej(p))$ is strictly concave on $[0,p_1]$ and on $[p_1,t]$.
\end{lemma}	

Note that $F^R$ will generally not be strictly concave over the combined interval $[0,t]$, due to a kink at $p_1$.

\begin{proof}
	We may write $R^\rej$ as:
	\begin{equation*}
		R^\rej(p)=\begin{cases}
			\mu_\rej R_{\leq t}(p)+(1-\mu_\rej)R_{\leq p_1}(p)&p\in[0,p_1]\\
			\mu_\rej R_{\leq t}(p)&p\in[p_1,t],
		\end{cases}
	\end{equation*}
where $\mu_\rej=\mu t/(\mu t+(1-\mu)p_1)$ denotes the conditional probability of a sophisticated buyer given a round-$1$ reject. By Lemma~\ref{lem:trunccon}, both $R_{\leq p_1}$ and $R_{\leq t}$ are strictly concave. Hence, $R^{\rej}$ is strictly concave in both cases.
\end{proof}

\begin{proof}[Proof of Lemma~\ref{lem:plph}]

As a consequence of Lemma~\ref{lem:pwconcave}, $R^\rej$ has exactly one local maximum on each of $[0,p_1]$ and $[p_1,t]$. 

Define $p_L$ and $p_H$ to be these respective maxima.
Note that for all $p\in [p_1,t]$, $R_{\leq p_1}(p)=0$. Thus, we can write
\begin{align*}
p_H= \arg \max_{p\in [p_1,t]}\mu_\rej R_{\leq t}(p)=p^*_{\leq t}.
\end{align*}

If $\mu=1$, which implies that $\mu_\rej=1$, then $R^\rej$ is strictly concave and maximized at $p^\rej_2=p^*_{\leq t}$. Since the equilibrium is sophisticated-focused we have that $p_1=p^\rej_2=p^*_{\leq t}$. Thus, the global optimal price $p^*_{\leq t}$ lies in $[0,p_1]$, hence $p_L=p^*_{\leq t}$.  

For $\mu<1$, we  prove that   $p^\rej_2$ cannot be deterministic. 
If it were, it would have to globally maximize $R^\rej$.
Moreover, it must be that $p_1=p_2^\rej$: for sophisticated-focused continuations, the threshold equation is $t-p_1=t-\mathbb E[p_2^\rej]$.
Hence we must have that $p_1$ globally maximizes $R^\rej$.
Since $R_{\leq p_1}$ is strictly concave and $R_{\leq p_1}(p_1)=0$, $R_{\leq p_1}'(p_1)<0$.
Hence, the left derivative of $R^\rej$ at $p_1$ is strictly less than its right derivative, i.e.
\begin{equation*}
    (1-\mu_R)R'_{\leq p_1}(p_1)+ \mu_R R'_{\leq t}(p_1)<\mu_R R'_{\leq t}(p_1).
\end{equation*}
This means it is impossible for $p_1$ to be locally or globally optimal.
Thus, $p^\rej_2$ must be randomized. 
This is only possible if it is supported on values both above and below $p_1$, and $p_L$ and $p_H$ are the only candidates.
\end{proof}

\section{Proof of Lemma~\ref{lem:p2adet}}
\label{app:a2det}

We first prove that $p_2^\acc$ is deterministic.
The seller's revenue for a price $p$ conditioned on accept can be written as $R^\acc(p)=\mu_\acc R_{\geq t}(p)+(1-\mu_\acc)R_{\geq p_1}(p)$, where for any $x$,
$$
R_{\geq x}(p)=\begin{cases}
	p&p< x\\
	p\tfrac{1-F(p)}{1-F(x)}&p\geq x
	\end{cases}.
$$
By Assumption~\ref{ass:regular} and the fact that $F$ is fully supported, both $ R_{\geq t}$ and $R_{\geq p_1}$ are continuous, and hence so too is $R^\acc$. 
Moreover, Assumption~\ref{ass:regular} implies that $R_{\geq t}$ and $R_{\geq p_1}$ are strictly concave above $t$ and  $p_1$, respectively.
Using the fact that $p_1\leq t$, we therefore, have that $R^\acc$ is strictly concave above $p_1$, and linear below.
These facts together imply that $R^\acc$ must have exactly one global maximum.
Since any randomization of $p_2^\acc$  must be between maxima of $R^\acc$, this implies that $p_2^\acc$  is deterministic.

Now assume that the continuation is sophisticated-focused, i.e.\ $p_2^\acc\geq t$.
We will show that $p_2^\acc=\max(p^*,t)$.
Note that for $p\geq t$, we can write:
\begin{align}
    R^\acc(p) &= \mu_\acc R_{\geq t}(p)+(1-\mu_\acc)R_{\geq p_1}(p)\notag\\
            &= \mu_\acc p\tfrac{1-F(p)}{1-F(t)} + (1-\mu_\acc) p\tfrac{1-F(p)}{1-F(p_1)}\notag\\
            &= \Big(\tfrac{\mu_\acc}{1-F(t)}+\tfrac{1-\mu_\acc}{1-F(p_1)}\Big)p(1-F(p)).\label{eq:rarewrite}
\end{align}
Note that the first factor in (\ref{eq:rarewrite}) does not depend on $p$.
It follows that $R^\acc(p)$ has the same maxima as $R$ on $p\geq t$.
Therefore if $p^*\geq t$, then $p_2^\acc=p^*$. 
Otherwise, by the concavity of $R$ and thus $R^\acc$, it must be that $p_2^\acc=t$.
Hence $p_2^\acc=\max(p^*,t)$.


\section{Proof of Lemma~\ref{lem:ttop}}
\label{app:ttop}

We divide Lemma~\ref{lem:ttop} into several pieces and prove each one separately.
First, we prove that given $\mu$, sophisticated-focused continuations are exactly those where $(1-\mu) R'(t)+(1-F(t))\mu\geq 0$.
Moreover, the choice of $p_2^\acc$ is uniquely determined.
\begin{lemma}\label{lem:sophp2a}
    Given $\mu$, a continuation $(p_1,p_2^\rej,p_2^\acc,t)$ is sophisticated if and only if $(1-\mu) R'(t)+(1-F(t))\mu\geq 0$.
    For such a continuation, it must be that $p_2^\acc=\max(p^*,t)$.
\end{lemma}
\begin{proof}
From the proof of Lemma~\ref{lem:p2adet}, we know that $R^\acc$ is strictly concave above $p_1$, and linear below.
Hence, $R^\acc$ has a single global maximum.
We will show that $(1-\mu) R'(t)+(1-F(t))\mu\geq 0$ if and only if the left derivative of $R^\acc$ at $t$ is positive, or equivalently the global maximum lies to the right of $t$.
For $p\geq t$,  $R^\acc(p)=(\tfrac{\mu_\acc}{1-F(t)}+\tfrac{1-\mu_\acc}{1-F(p_1)})p(1-F(p))$, and hence $R^\acc(p)$ is uniquely maximized at whichever is larger of $p^*$ or $t$.

To compute the left derivative of $R^\acc$ at t, note that $R^\acc(p)$ is
$$
(1-\mu_\acc)R_{\geq p_1}(p)+ \mu_\acc R_{\geq t}(p)=
(1-\mu_\acc)R_{\geq p_1}(p)+ \mu_\acc p.
$$
The left derivative of $R^\acc$ with respect to $p$ at $t$ is therefore given by
\begin{equation}
    \frac{d^- R^\acc(p)}{d p} \Big |_t=
(1-\mu_\acc)\frac{R'(t)}{1-F(p_1)}+ \mu_\acc.
\label{eq:leftderiv}
\end{equation}
Recall the formulas 
\begin{align*}
\mu_\acc = \frac{(1-F(t))\mu}{(1-F(p_1))(1-\mu)+(1-F(t))\mu}&&(1-\mu_\acc)=\frac{(1-F(p_1))(1-\mu)}{(1-F(p_1))(1-\mu)+(1-F(t))\mu}.
\end{align*}
Substituting into (\ref{eq:leftderiv}), we obtain:
\begin{align*}
\frac{d^-R^\acc(t)}{dp} &=
\frac{(1-F(p_1))(1-\mu)}{(1-F(p_1))(1-\mu)+(1-F(t))\mu}\cdot\frac{R'(t)}{(1-F(p_1))}+\frac{(1-F(t))\mu}{(1-F(p_1))(1-\mu)+(1-F(t))\mu} \\
 &=
\frac{(1-\mu)R'(t)}{(1-F(p_1))(1-\mu)+(1-F(t))\mu}+\frac{(1-F(t))\mu}{(1-F(p_1))(1-\mu)+(1-F(t))\mu} \\
 &=
\frac{(1-\mu)R'(t)+(1-F(t))\mu}{(1-F(p_1))(1-\mu)+(1-F(t))\mu}.
\end{align*}
The denominator is always positive, so the left derivative of $R^\acc $ is non-negative if and only if $(1-\mu)R'(t)+(1-F(t))\mu\geq 0$. 
\end{proof}

We now show how to find a a sophisticated-focused implementation for any $t$ satisfying $(1-\mu) R'(t)+(1-F(t))\mu\geq 0$, and prove its uniqueness.

\begin{lemma}\label{lem:raderiv}
    Given $t$ and $\mu$ such that $(1-\mu) R'(t)+(1-F(t))\mu\geq 0$, there exists a unique $p_1$, $p_2^\acc$, and $p_2^\rej$ that implement $t$ in a sophisticated-focused continuation.
\end{lemma}
\begin{proof}
Given $t$ and $\mu$, we will find $p_1$ and a randomized $p_2^\rej$ such that $p_1=\mathbb E[p_2^\rej]$, and such that all prices in the support of $p_2^\rej$ maximize $R^\rej$.
We show how to do so below, and in the process, it will be clear that these choices are unique.
Given $p_1$, note that $p_2^\acc=\max(p^*,t)$ is the only choice for accept price.

We now show how to find $p_1$ and the distribution for $p_2^\rej$.
By Lemma~\ref{lem:pwconcave},  we know $R^\rej$ is strictly concave on both $[0,p_1]$ and $[p_1,t]$.
Hence the only candidate global optima for $R^\rej$ are the local optima for its two pieces, given by:
\begin{itemize}
		\item $p_H=p^*_{\leq t}$: the monopoly price for the truncated distribution $F_{\leq t}$.
		\item $p_L$: the maximum value of $\mu_\rej R_{\leq t}(p)+(1-\mu_\rej)R_{\leq p_1}(p)$ on $[0,p_1]$.
\end{itemize}
Hence if we find $p_1$ such that $R^\rej(p_L)=R^\rej(p_H)$, these points are guaranteed to be globally optimal, and we may choose probabilities of $p_L$ and $p_H$ such that $p_1=\mathbb E[p_2^\rej]$.
We show how to find such a $p_1$ below.
Note that given a $p_1$ satisfying $R^\rej(p_L)=R^\rej(p_H)$, the uniqueness of the distribution for $p_2^\rej$ follows immediately.

To identify a $p_1$ such that $R^\rej(p_L)=R^\rej(p_H)$.
First, that $R^\rej(p_H)=\mu_\rej R_{\leq t}(p^*_{\leq t})$ is continuous and strictly decreasing in $p_1$: the only factor that depends on $p_1$ is $\mu_\rej$, which is continuous and strictly decreasing as a function of $p_1$ in the interval $[0,p^*_{\leq t})$.
Meanwhile, $R^\rej(p_L)$ is continuous and strictly increasing in $p_1$.
This requires more justification, as $p_L$ is defined as the unique maximum of $\mu_\rej R_{\leq t}(p)+(1-\mu_\rej)R_{\leq p_1}(p)$ on $[0,p_1]$.
To prove the claim, it suffices to show that for any $p\in [0,p_1]$, the revenue $\mu_\rej R_{\leq t}(p)+(1-\mu_\rej)R_{\leq p_1}(p)$ is strictly increasing and continuous in $p_1$.
We can write this quantity as:
\begin{align*}
    &\frac{\mu F(t)}{\mu F(t)+(1-\mu)F(p_1)}p\left(1-\frac{F(p)}{F(t)}\right)+\frac{(1-\mu) F(p_1)}{\mu F(t)+(1-\mu)F(p_1)}p\left(1-\frac{F(p)}{F(p_1)}\right)\\
    &\quad\quad\quad\quad=p\left(1-\frac{F(p)}{\mu F(t)+(1-\mu)F(p_1)}\right).
\end{align*}
This quantity is strictly increasing in $p_1$.

To prove the existence and uniqueness of $p_1$, then, notice that when $p_1=0$, $R^\rej(p_L)=0$ while $R^\rej(p_H)=\mu_\rej R_{\leq t}(p^*_{\leq t})$ is positive.
Meanwhile at $p_1=p_H$, we have $R^\rej(p_L)=\max_{p\leq p_1}R^\rej(p)=max_{p\leq p_H}R^\rej(p)\geq R^\rej(p_H)$.
By continuity, there must therefore exist some $p_1$ inducing equality, and by the strict monotonicity of both functions this point must be unique.
\end{proof}


\section{Proof of Lemma~\ref{lem:p1monotone}}
\label{app:p1monotone}

Let $t$ and $\mu$ admit a sophisticated-focused continuation with first round price $p_1$.
We can write the difference of $p_L$ and $p_H$ as:
\begin{align*}
R^\rej(p_L)- R^\rej(p_H) 
 &= 
 \frac{\mu p_L(F(t)-F(p_L))+(1-\mu)p_L(F(p_1)-F(p_L))-\mu p^*_{\leq t}(F(t)-F(p^*_{\leq t}))}{\mu F(t)+(1-\mu)F(p_1)}.
\end{align*}

Since $p_1$ and $\mu$ implement a sophisticated-focused continuation at $t$ the difference is $R^\rej(p_L)- R^\rej(p_H)=0$. Thus, it must be that the numerator 
$$
\mu p_L(F(t)-F(p_L))+(1-\mu)p_L(F(p_1)-F(p_L))-\mu p^*_{\leq t}(F(t)-F(p^*_{\leq t})=0.$$
Solving for $p_1$ as function of $\mu$:
\begin{align*}
p_1&=F^{-1} \left( -\frac{\mu p_L (F(t)-F(p_L))-(1-\mu)p_LF(p_L)- \mu p^*_{\leq t}(F(t)-F(p^*_{\leq t})}{(1-\mu)p_L}               \right)\\
&=F^{-1} \left( -\frac{\mu p_L (F(t)-F(p_L))- \mu p^*_{\leq t}(F(t)-F(p^*_{\leq t})}{(1-\mu)p_L}
+F(p_L)
\right)\\
&=F^{-1} \left( -\frac{\mu R_{\leq t}(p_L)F(t)- \mu R_{\leq t}(p^*_{\leq t})}{(1-\mu)p_L}
+F(p_L)\right)\\
&=F^{-1} \left(\frac{\mu}{1-\mu} \frac{F(t) (R_{\leq t}(p^*_{\leq t})-R_{\leq t}(p_L))}{p_L}
+F(p_L)
\right).
\end{align*}
Thus, we can verify that $p_1$ is continuous and differentiable with respect to $\mu$. Moreover, $F^{-1}(x)$ is strictly increasing in $x$ by assumption, $\frac{\mu}{1-\mu}$ is strictly increasing in $\mu$ and $(R_{\leq t}(p^*_{\leq t})-R_{\leq t}(p_L))>0$ since $p^*_{\leq t}$ is the unique optimal of $R_{\leq t}$. Thus, $p_1(t,\mu)$ is strictly increasing in $\mu$. 

\section{Proof of Lemma~\ref{lem:p1nonaive}}
\label{app:p1nonaive}

For the sake of contradiction, assume there are two continuations for $p_1$: 
 \begin{itemize}
        \item Sophisticated-focused: A randomized $p_2^{\rej,s}$ and $t^s\geq p_2^\acc$.
        \item Naive-focused continuation; A possibly randomized $p_2^{\rej,n}$ and $p_2^{\acc,n}<t^n$
    \end{itemize}
First, we argue that if both continuations are possible, it must be that $t^n>t^s$:
Lemma~\ref{lem:ttop} implies that if $t^n\leq t^s$, then it can only be implemented in a sophisticated-focused continuation.

Next, we show that $p^{\rej,s}_2\leq p^{\rej,n}_2$ in all realizations.
Lemma~\ref{lem:plph} implies that $p^{\rej,s}_2$ is a random price supported on the two distinct global maxima, and that the higher of the two is  $p^*_{\leq t^s}$, the monopoly reserve of the truncated distribution $F_{\leq t^s}$. This implies that for every $p\leq p^*_{\leq t^s}$:
\begin{align}
\label{eq:ptsopt}
R^{\rej,s}(p)=(1-\mu_\rej(t^s))R_{\leq p_1}(p)+\mu_\rej(t^s)R_{\leq t^s}(p)
&\leq
\mu_\rej(t^s)R_{\leq t^s}(p^*_{\leq t^s}) =
 R^{\rej,s}(p^*_{\leq t^s})
\end{align}
We may therefore obtain the following sequence of relations, explained after their statement:
 \begin{align}
 R^{\rej,n}(p)&=(1-\mu_\rej(t^n))R_{\leq p_1}(p)+\mu_\rej(t^n)R_{\leq t^n}(p)\notag\\
 &\leq  (1-\mu_\rej(t^s))R_{\leq p_1}(p)+\mu_\rej(t^n)R_{\leq t^n}(p)\label{eq:1}\\
&\leq  (1-\mu_\rej(t^s))R_{\leq p_1}(p)+\mu_\rej(t^n)p\left(1-\frac{F(p)}{F(t^n)}\right)\label{eq:2}\\
&= (1-\mu_\rej(t^s))R_{\leq p_1}(p)+\mu_\rej(t^n)p\left(1-\frac{F(p)}{F(t^s)} + \frac{F(p)}{F(t^s)}-\frac{F(p)}{F(t^n)}\right) \notag\\
&= (1-\mu_\rej(t^s))R_{\leq p_1}(p)+\mu_\rej(t^n)R_{\leq t^s}(p) +\mu_\rej(t^n)p\left(\frac{F(p)}{F(t^s)}-\frac{F(p)}{F(t^n)}\right)\label{eq:3}\\
&=
R^{\rej,s}(p)+(\mu_\rej(t^n)-\mu_\rej(t^s))R_{\leq t^s}(p)+\mu_\rej(t^n)p\left(\frac{F(p)}{F(t^s)}-\frac{F(p)}{F(t^n)}\right)\label{eq:4}\\
&=
R^{\rej,s}(p)+(\mu_\rej(t^n)-\mu_\rej(t^s))R_{\leq t^s}(p)+\mu_\rej(t^n)pF(p)\left(\frac{1}{F(t^s)}-\frac{1}{F(t^n)}\right)\notag\\
&\leq
R^{\rej,s}(p^*_{\leq t^s})+(\mu_\rej(t^n)-\mu_\rej(t^s))R_{\leq t^s}(p^*_{\leq t^s})\notag\\
&\quad\quad\quad\quad+\mu_\rej(t^n)p^*_{\leq t^s}F(p^*_{\leq t^s})\left(\frac{1}{F(t^s)}-\frac{1}{F(t^n)}\right)\label{eq:5}\\
&= R^{\rej,s}(p^*_{\leq t^s})+(\mu_\rej(t^n)-\mu_\rej(t^s)p^*_{\leq t}
\left(1-\frac{F(p^*_{\leq t})}{F(t^s)}\right)\notag\\
&\quad\quad\quad\quad+\mu_\rej(t^n)p^*_{\leq t^s}F(p^*_{\leq t^s})\left(\frac{1}{F(t^s)}-\frac{1}{F(t^n)}\right)\label{eq:6}\\
&= R^{\rej,s}(p^*_{\leq t^s})-\mu_\rej(t^s)p^*_{\leq t}
\left(1-\frac{F(p^*_{\leq t})}{F(t^s)}\right)
+\mu_\rej(t^n)p^*_{\leq t^s}\left(1-\frac{F(p^*_{\leq t})}{F(t^n)}\right)\notag\\
&=
\mu_\rej(t^n)p^*_{\leq t^s}(1-\frac{F(p^*_{\leq t})}{F(t^n)})\label{eq:6.5}
\\
&=R^{\rej,n}(p^*_{\leq t^s})\label{eq:7}
\end{align}

Line (\ref{eq:1}) follows from the fact that $1-\mu_\rej(t^s)\geq 1-\mu_\rej(t^s)$ since $\mu_\rej()$ is strictly decreasing in the threshold. 
Line (\ref{eq:2}) comes from the definition of $R_{\leq t^s}(p)=p(1-\frac{F(p)}{F(t^s)})$.
Line (\ref{eq:3}) comes from the definition of $R^{\rej,s}(p)$. 
Line (\ref{eq:4}) follows from the definition of $R^{\rej,s}(p)$.
Line (\ref{eq:5}) comes from the following facts:
\begin{itemize}
    \item $R^{\rej,s}(p)\leq R^{\rej,s}(p^*_{\leq t^s})$ (Equation (\ref{eq:ptsopt})).
    \item  $R_{\leq t^s}(p)\leq R_{\leq t^s}(p^*_{\leq t^s})$, by the definition of $p^*_{\leq t^s}$ it is the monopoly reserve of $F_{\leq t^s}$.
    \item $pF(p)$ is increasing in $p$ and $p\leq p^*_{\leq t^s}$.
\end{itemize}
Line (\ref{eq:6}) uses the definition of $R_{\leq t^s}$.
Line (\ref{eq:6.5}) follows from the definition of $R^{\rej,s}(p^*_{\leq t^s})$.
Finally, line (\ref{eq:7}) uses the definition of $R^{\rej,s}(p^*_{\leq t^s})=\mu_\rej(t^s)R_{\leq t^s}(p^*_{\leq t^s})=\mu_\rej(t^s)p^*_{\leq t^s}(1-\frac{F(p^*_{\leq t^s})}{F(t^s)})$. As a result, $p_2^{\rej,n}\geq p^*_{\leq t^s}$ for every realization of $p_2^{\rej,n}$.

Note that Lemma~\ref{lem:plph} characterizes the sophisticated-focused continuation such that $p_1=E[p^{\rej,s}_2]$ and $p^{\rej,s}_2$ is not deterministic. Thus it must be that
$p_1=E[p^{\rej,s}_2]<p^*_{\leq t^s}$, which implies that
\begin{equation}
\label{eq:p2rinc}
    p_1<p^{\rej,s}_2\leq p^{\rej,n}_2
\end{equation}
We conclude our proof using the threshold equation for the naive-focused continuation: $t^s-p_1+t^s-p^{\acc,s}_2 = t^s-p^{\rej,s}_2$ and hence $
  t^s+p^{\rej,s}_2 = p_1+p^{\acc,s}_2$.
Since $p^{\acc,s}<t$, it must be that $p_1>p^{\rej,s}_2$, which contradicts Equation (\ref{eq:p2rinc}).

\section{Proof of Lemma~\ref{lem:trev}}
\label{app:trev}

The seller's revenue for implementing $t$ can be decomposed as follows:
\begin{align*}
\Rev_1&=R(p_1)(1-\mu) + p_1(1-F(t))\mu\\
\Rev_2^\acc&=\big(\mu(1-F(t))+(1-\mu)(1-F(p_1)\big)R^\acc(p_2^\acc)\\
\Rev_2^\rej&=\big(\mu(F(t))+(1-\mu)(F(p_1)\big)\mathbb E[R^\rej(p_2^\rej)].
\end{align*}

From the characterization of Lemma 4.9, we know that $p^\rej_2$ is a randomized price supported on two prices $p_L<p_1$ and $p_1\leq p_H=p^*_{\leq t}$ such that $\mathbb E[R^\rej(p_2^\rej)]=R^\rej(p_L)=R^\rej(p^*_{\leq t})$. This implies that the revenue from offering the randomized price $p_2^\rej$ is the same as if we offered price $p^*_{\leq t}$. Hence we may rewrite $\Rev_2^\rej$ to obtain the following reamortization:
$$
\Rev_2^\rej=(\mu F(t)+(1-\mu )F(p_1)) R^\rej(p^*_{\leq t})
$$
Note that since $p^*_{\leq t}\geq p_1$, $p^*_{\leq t}$ lies in the support of $F_{\leq t}$ but not $F_{\leq p_1}$.
Thus, $$R^\rej(p^*_{\leq t})= \mu_\rej \frac{F(t)-F(p^*_{\leq t})}{F(t)}p^*_{\leq t} =
\frac{\mu (F(t)-F(p^*_{\leq t}))}{(\mu F(t)+(1-\mu)F(p_1))}p^*_{\leq t} .
$$
Recall $\mu_\rej$ is the probability that the agent is sophisticated given a reject.  Thus, we obtain:
$$
\Rev_2^\rej=(\mu F(t)+(1-\mu )F(p_1)) R^\rej(p^*_{\leq t}) = \mu (F(t)-F(p^*_{\leq t}))p^*_{\leq t}.
$$
Furthermore, using the definition of $R^\acc$, we can simplify $\Rev_2^\acc$ to obtain $\Rev_2^\acc=(1-F(p_2^\acc))p_2^\acc$.

Thus the total revenue is given by:
\begin{align*}
\Rev&=\Rev_1+\Rev_2^\rej+\Rev_2^\acc\\
&=  (1-\mu)R(p_1)+\mu \big(p_1(1-F(t))+p^*_{\leq t}(F(t)-F(p^*_{\leq t}))\big) +(1-F(p_2^\acc))p_2^\acc\\
&=  R(p_1)+\mu \big(p_1(1-F(t))+p^*_{\leq t}(F(t)-F(p^*_{\leq t}))-(1-F(p_1))p_1\big) +(1-F(p_2^\acc))p_2^\acc
\\
&=  R(p_1)+\mu \big(p^*_{\leq t}(F(t)-F(p^*_{\leq t}))-(F(t)-F(p_1))p_1\big) +(1-F(p_2^\acc))p_2^\acc.
\end{align*}
Note that $p_2^\acc=\max(p^*,t)$ does not depend on $\mu$.
Moreover, for a fixed $\mu$ and $t$ admitting sophisticated-focused implementation, the choice of $p_1$ is unique, and moreover from Lemma~\ref{lem:p1monotone}, $p_1$ is continuous in $\mu$, thus, $\Rev$ is also continuous in $\mu$.   
Moreover, $p_1$ is differentiable, and strictly increasing in $\mu$. Letting $p_1'>0$ denote the derivative of this unique $p_1$ with respect to $\mu$,  we may write the derivative of the revenue with respect to $\mu$ as:
\begin{align*}
    (1-F(p_1))p'_1-p_1f(p_1)p'_1 +(F(t)-F(p^*_{\leq t}))p^*_{\leq t}-(F(t)-F(p_1))p_1  \\-\mu((F(t)-F(p_1))p'_1-p_1f(p_1)p_1').
    \end{align*}

Since $\rev_{\leq t}$ is strictly concave, it must be that $p_{\leq t}^*$ is the unique maximizer of $R_{\leq t}$. Hence:
$$(F(t)-F(p^*_{\leq t}))p^*_{\leq t}=R_{\leq t}(p^*_{\leq t})F(t) > R_{\leq t}(p_1)F(t)= (F(t)-F(p))p_1. $$
Thus, we can lower bound the derivative of the revenue as:
\begin{align*}
    &(1-F(p_1))p'_1- p_1f(p_1)p'_1-\mu(F(t)-F(p_1))p'_1+\mu p_1f(p_1)p'_1\\
    &\geq (1-\mu) (1-F(p_1))p'_1- p_1f(p_1)p'_1+\mu p_1f(p_1)p'_1\\
    &= (1-\mu) (1-F(p_1))p'_1 -(1-\mu)p_1f(p)p'_1\\
    &= (1-\mu)( (1-F(p_1))p'_1- p_1f(p)p'_1)\\
    &= (1-\mu)R'(p_1).
    \end{align*}
where the inequality follows from the fact that 
$\mu (1-F(p_1))p'_1 \geq \mu (F(t)-F(p_1))p'_1 $. Finally, we argue that $R'(p_1)>0$ or equivalently $R(p_1)$ increases with $\mu$: from Lemma~\ref{lem:p1monotone} we know that $p_1$ increases with $\mu$. 
Lemma~\ref{lem:ordering} implies that $p_1=\mathbb E[p_2^\rej]<p^*$, and hence $R'(p_1)>0$.

\section{Proof of Theorem~\ref{thm:robust}}
\label{app:robust}

Theorem~\ref{thm:robust} will follow from three observations. 
First, as $\mu$ increases, thresholds with sophisticated-focused implementations stay sophisticated-focused, and their revenue increases continuously.
Second, the full-sophisticated equilibrium is sophisticated-focused, and the revenue is high for the on-path threshold, even at lower levels of $\mu$.
Third, as $\mu$ increases, naive-focused continuations all eventually yield less revenue than the full-sophistication equilibrium.
We begin with the first observation.

\begin{lemma}\label{lem:bluebound}
Consider threshold $t<1$ and let $\Rev(t,\mu)$  be the revenue of the unique sophisticated-focused implementation of $t$ with sophistication level $\mu$, if one exists.
Then for any $x>0$, there exists a $\mu(x)$ such that for all $\mu>\mu(x)$, threshold $t$ has a sophisticated-focused implementation and 
$$
\Rev(t,1)-\Rev(t,\mu)<x.
$$
\end{lemma}

\begin{proof}
By Lemma~\ref{lem:ttop}, there exists a $\mu(t)$ such that $t$ admits a sophisticated-focused continuation.
Furthermore, any $\mu>\mu(t)$ implements a sophisticated-focused continuation. 
From Lemma~\ref{lem:trev} the revenue is continuous and strictly increasing in $[\mu(t),1]$. Thus, there must be some $\mu(x)\in [\mu(t),1)$ such that $Rev(t,\mu(x))+x\geq Rev(t,1)$. Since $\Rev$ is strictly increasing $\mu$, the result follows.
\end{proof}

Next, we show that the revenue at $\mu=1$ is bounded away from the single round Myerson revenue. Together with the previous lemma, this will lower bound the revenue of the best sophisticated-focused continuation at high enough $\mu$.

\begin{lemma}\label{lem:devrev}
	There exists a distribution-specific constant $\epsilon_F>0$ such that $\Rev(1)=R^*+\epsilon_F$.
\end{lemma}
\begin{proof}
	From the analysis of \citet{dps19} the perfect Bayesian equilibrium revenue with $\mu=1$ maximizes $R(t)+R(p^*_{\leq t})$ over all $t$. Since $F$ is fully supported, the conditional distribution $F_{\leq t}$ is fully supported on $[0,t]$, and hence $R(p^*_{\leq t})>0$ for all $t$. The result then follows from choosing $t$ to be the monopoly reserve for $F$, i.e.\ $t=p^*$.
\end{proof}

 The next step is to upper bound the revenue of naive-focused continuations when $t$ is sufficiently high. We will see that no other $t$ will support naive-focused continuations.

\begin{lemma}\label{lem:naiverev}
	Let $(p_1,p_2^\rej, p_2^\acc)$ be a naive-focused implementation of a threshold $t$ such that $F(t)\in (1-\epsilon_F/6,1]$. Then
		$$\Rev(p_1,p_2^\rej, p_2^\acc,t)\leq R^*+\epsilon_F/3+2(1-\mu).$$
\end{lemma}
\begin{proof}
	For any fixed $\mu$, we have 
	\begin{equation}\label{eq:greenrev}
		\Rev(p_1,p_2^\rej, p_2^\acc,t)=\mu\Rev_S(p_1,p_2^\rej, p_2^\acc,t)+(1-\mu)\Rev_N(p_1,p_2^\rej, p_2^\acc,t),
	\end{equation}
	where $\Rev_S$ and $\Rev_N$ are the expected revenues from the sophisticated and naive buyers, respectively. Note that both revenues are at most $2$, as no sequentially rational price will ever exceed $1$. In particular, the second term in (\ref{eq:greenrev}) is at most $2(1-\mu)$.
	
	Examining the revenue from sophisticated buyers in more detail, we have:
	\begin{equation}\label{eq:sophgreenrev}
		\Rev_S(p_1,p_2^\rej, p_2^\acc,t)=p_1(1-F(t))+p_2^\rej(F(t)-F(p_2^\rej))+p_2^\acc(1-F(t)).
	\end{equation}
	The first term in (\ref{eq:sophgreenrev}) is the revenue from the first round, where sophisticated buyers above $t$ buy. The second term is the revenue from the reject case in round $2$. The final term comes from the accept case. Note that conditioned on an acceptance in round $1$, all sophisticated buyers accept in round $2$ as well, because we have assumed our implementation of $t$ is naive-focused, i.e.\ $p_2^\acc<t$.
 
 We can now bound each term in (\ref{eq:sophgreenrev}) individually. The first and third terms are at most $\epsilon_F/6$, since $p_1,p_2^\acc\leq 1$ and $F(t)> 1-\epsilon_F/6$. The second term is at most $R^*_{\leq t}$. To see this, note that $p_2^\rej$ maximizes the seller's revenue from a mixture of sophisticated and naive buyers. Hence, the seller will obtain less revenue from the sophisticated buyers than if they offered the optimal price for that population, which has distribution $F_{\leq t}$. Hence, $p_2^\rej(F(t)-F(p_2^\rej))/F(t)\leq R^*_{\leq t}$. Moreover, $R^*_{\leq t}\leq R^*$, and we have $F(t)\leq 1$, upperbounding the second term as $R^*$ overall. We can therefore bound the sophisticated revenue as $\Rev_S(p_1,p_2^\rej, p_2^\acc,t)\leq R^*+\epsilon_F/3$. Noting that $\mu\leq 1$, we can combine our naive and sophisticated revenue bounds to obtain the bound stated in the lemma.
\end{proof}

Finally, we are ready to proceed with the proof of Theorem~\ref{thm:robust}:

\begin{proof}[Proof of Theorem~\ref{thm:robust}]
Let $t^*$ be the equilibrium threshold for $\mu=1$. Lemma~\ref{lem:devrev} implies that
$$
Rev(t^*,1) \geq R^*+\epsilon_F.
$$

Let $t$ be such that $F(t)=1-\epsilon_F/6$. Lemma~\ref{lem:ttop} guarantees that there exists a sophistication level $\mu_1$ such that  for all $\mu >\mu_1$ any implementation of $t$ must be sophisticated-focused. From Lemma~\ref{lem:ttop} this also holds for smaller values of $t$. Thus, if $\mu >\mu_1$ any naive-focused continuation must use $t$ such that $F(t)>1-\epsilon_F/6$. Thus, from Lemma~\ref{lem:naiverev} the revenue of any naive-focused continuation satisfies
\begin{align*}
Rev(t,\mu)\leq R^*+\epsilon_F/6+2(1-\mu) && \forall \mu\in  (\mu_1,1] .   
\end{align*}

Finally, Lemma~\ref{lem:bluebound} implies that there is a $\mu_2\geq \mu(t^*)$ such that 
\begin{align*}
Rev(t^*,\mu) =R^*+\epsilon_F/2     &&\mu \in (\mu_2,1].
\end{align*}
Let $\mu_3 = \max \{\mu_1,\mu_2\}$. We conclude that for all $\mu>\hat \mu$ and all $t$ with naive-focused implementations:
\begin{align*}
 Rev(t^*,\mu)-Rev(t,\mu) >\epsilon_F/3 - 2(1-\mu) && \forall \mu\in (\mu_3,1).
\end{align*}
We conclude the proof by using 
$\hat \mu = \max(\mu_3, 1-\epsilon_F/6 )$.
\end{proof}

\section{Comparative Analyses}
\label{app:comparative}

In this appendix, we provide greater intuitive clarity on the causes of the seller's revenue gains from sophisticated buyers. 
We provide theoretical evidence that the reason the seller's revenue increases because removing naive buyers improves the seller's ability to bargain.
In other words, the presence of naive buyers gives sophisticated buyers leverage by making their demand reduction more credible.
This distorts the seller's prices; reducing the number of naive buyers reduces the distortion.

We provide three pieces of evidence for this interpretation.
In Sections~\ref{sec:firstorder} and \ref{sec:welfare}, we rule out two alternative explanations for the seller's increased revenue.
We show, respectively, that it is not the case that sophisticated buyers are intrinsically more lucrative, or that the equilibrium welfare increases as sophisticated buyers are added.
Then in Section~\ref{sec:stackelberg}, we give evidence that the seller's lack of commitment is crucial.
We analyze the same model, but give the seller the power to commit to a deterministic price schedule.
In these circumstances, adding sophisticated buyers only hurts the seller's revenue.

\subsection{Population-Driven Revenue Effects}
\label{sec:firstorder}

As the population sophistication $\mu$ increases, the seller's revenue change can be broken into two types.
{Population-driven effects} account for the revenue gain from new sophisticated buyers, and loss from losing naive buyers:
if one population gives more revenue than the other, on average, then changing how many of each show up will change the seller's revenue.
{Price-driven effects} account for the change in revenue from the existing populations of naive and sophisticated buyers.
As the populations change, so does the structure of the price schedule.

We show that population-driven effects can't possibly explain the effect in Theorem~\ref{thm:2rev}. 
In particular, we show that for general distributions that the seller makes weakly more revenue from a naive buyer, compared to a sophisticated buyer.
Hence, the phenomenon of Theorem~\ref{thm:2rev} is not caused by population-driven effects. 

\begin{theorem}
\label{thm:firstorder}
    For any sophistication level $\mu$ and corresponding PBE, the seller earns more revenue conditioned on the buyer being naive than conditioned on the buyer being sophisticated.
\end{theorem}

\begin{proof}
Let $(p_1,p_2^\rej, p_2^\acc,t)$ be the on-path continuation.
From the buyer indifference equation and Lemma~\ref{lem:ordering}, it follows that $p_1-\mathbb E[p_2^\rej]=\mathbb E[(t-p_2^\acc)^+]\geq 0$, and hence that $p_1\geq \mathbb E[p_2^\rej]$. Now conditioned on the buyer being naive, the revenue is
$$
R_{n}=p_1(1-F(p_1))+\mathbb E[p_2^\acc(1-F(p_2^\acc))]+\mathbb E[p_2^\rej(F(p_1)-F(p_2^\rej))\mathbb I(p_2^\rej\leq p_1)]
$$
and conditioned on a sophisticated buyer, we have revenue
$$
R_{s}=p_1(1-F(t))+\mathbb E[p_2^\acc(1-F(p_2^\acc)\mathbb I(p_2^\acc\leq t)]+\mathbb E[p_2^\rej(F(t)-F(p_2^\rej))].
$$
Noting that $p_1\geq \mathbb E[p_2^\rej]$ and $\mathbb E[p_2^\acc(1-F(p_2^\acc))]\geq \mathbb E[p_2^\acc(1-F(p_2^\acc))\mathbb I(p_2^\acc\leq t)]$, we may lower bound the revenue difference $R_n-R_s$ as:
\begin{align*}
 &\mathbb E[p_2^\rej(F(t)-F(p_1))]+\mathbb E[p_2^\rej(F(p_1)-F(p_2^\rej))\mathbb I(p_2^\rej\leq p_1)]-\mathbb E[p_2^\rej(F(t)-F(p_2^\rej))]\\
 &=\mathbb E[p_2^\rej(F(p_2^\rej)-F(p_1))\mathbb I(p_2^\rej>p_1)]\geq 0.
\end{align*}
\end{proof}

\subsection{Welfare and Bargaining}
\label{sec:welfare}

Another alternative explanation for the seller's revenue gains near full sophistication is that the welfare might increase with the population sophistication $\mu$.
With greater surplus to split between the seller and buyers, it is possible for the seller to profit without bargaining better.
This turns out not to be the case for the linear-demand equilibrium of Section~\ref{sec:uniform}. 
For the region where revenue increases with $\mu$ ($\mu\geq 0.630...$), the welfare decreases with $\mu$.

In more detail, the seller's equilibrium revenue in the region $\mu\in [0.630...,1]$ is given by the formula:
\begin{equation}\label{eq:welfare}
\text{Welf}=(1-\mu)\left(\tfrac{1-p_1}{2}+\tfrac{1-p_2^\acc}{2}+\mathbb E[\tfrac{(p_1-p_2^\rej)^+}{2}]\right)+\mu\left(1-t+\tfrac{t-p_1}{2}\right).
\end{equation}
The first term comes from naive buyers, and the second from sophisticated buyers.
The first of the naive terms comes from the round $1$ welfare. The second and third come from the round $2$ welfare. 
The first of the sophisticated buyers terms comes from the buyers who purchase two units (who have value in the range $[t,1]$), and the second from buyers who purchase only one unit (using the fact that $\mathbb E[p_2^\rej]=p_1$.
Plugging the formulas from Appendix~\ref{app:uniform} into (\ref{eq:welfare}), we obtain:
\begin{align*}
    \text{Welf}(\mu)=\tfrac{6+9\sqrt\mu+6\mu+\mu^{3/2}}{8+16\sqrt\mu+14\mu+4\mu^{3/2}-2\mu^2}&\quad\quad\quad\mu\geq 0.630...,
\end{align*}
which is a decreasing function.

One interpretation of this result is that naive buyers make the Coasian equilibrium ``more Coasian.'' For $\mu\geq 0.630...$, the price structure resembles the $\mu=1$ case derived by \citet{dps19}: the seller offers low prices, and ends up giving the item for less than they would prefer, more often than they prefer.
Adding naive buyers to this equilibrium causes prices to go even lower, and causes the seller to sell the item even more often.

\subsection{Heterogeneous Buyers with Commitment}
\label{sec:stackelberg}

In this section, we prove that the revenue gain observed in Theorem~\ref{thm:2rev} is unique to the no-commitment setting.
We study a model where the seller has commitment power, but which is otherwise identical to the two-period model in the rest of the paper.
With commitment power, the seller's strategy can be summarized by a price schedule of $p_1$, $p_2^\rej$, and $p_2^\acc$.
Importantly, $p_2^\rej$ and $p_2^\acc$ no longer need to satisfy sequential rationality, and hence can be suboptimal given the seller's updated round-$2$ beliefs.
Given a choice of $p_1$, $p_2^\rej$, and $p_2^\acc$, naive buyers still behave as price-takers.
A sophisticated buyer's strategy depends on the relative values of the prices.
A sophisticated buyer may buy $0$, $1$, or $2$ units.
To buy $2$ units, they pay $p_1+p_2^\acc$.
To buy $1$ unit, they pay $\min(p_1,p_2^\rej)$.
Their strategy, then can be computed by comparing their utility from the three options, given by $0$, $v-\min(p_1,p_2^\rej)$, and $2v-p_1+p_2^\acc$, respectively. We will show that the seller earns less revenue as sophistication grows. To do so, we show that there is a price schedule with the natural ordering $p_2^\rej\leq p_1\leq p_2^\acc$. For prices ordered this way, the result then follows fairly straightforwardly. We first order $p_1$ and $p_2^\acc$, before pinning down the full ordering.

\begin{lemma}
\label{lem:p1&p2a}
There exists a revenue-optimal price schedule with $p_1 \leq p_2^\acc$.
\end{lemma}

\begin{proof}
Consider a price schedule where $p_1>p_2^\acc$. 
We will modify the prices to produce a new schedule with weakly greater revenue.
We consider three cases, based on the location of $p_2^\rej$.

\textbf{Case 1:} $p_2^\rej \leq p_2^\acc < p_1$. 
Consider swapping $p_2^\acc$ and $p_1$.
We receive the same revenue from sophisticated buyers: the cost for $1$ and $2$ units remains the same, so their behavior remains unchanged. For na\"ive buyers, we receive weakly more revenue. If a na\"ive buyer rejected $p_1$, they may accept $p_2^A$ in the first round now. If they accepted $p_1$, they will definitely accept $p_2^A$ in the first round. Thus, in this case, we have shown a new equilibrium where $p_2^R \leq \hat{p}_1 < \hat{p}_2^A$ (note $\hat{p}_1 = p_2^A$ and $\hat{p}_2^A = p_1$ are the updated prices).\\
\textbf{Case 2:} In the case where $p_2^A < p_1 \leq p_2^R$, consider swapping $p_2^A$ and $p_1$ again. Observe that the seller still receives the same amount of revenue from the sophisticated buyers. For na\"ive buyers, we receive weakly more revenue. If a na\"ive buyer rejected $p_1$, they definitely rejected $p_2^R$, so we may obtain some revenue if they accept $p_2^A$ in the fisrt round. If they accepted $p_1$, they will definitely accept $p_2^A$ in the first round. Thus, in this case, we have shown a new equilibrium where $\hat{p}_1 < \hat{p}_2^A \leq p_2^R$ (note $\hat{p}_1 = p_2^A$ and $\hat{p}_2^A = p_1$ are the updated prices).\\
\textbf{Case 3:} In the case where $p_2^R \in (p_2^A, p_1)$ with $p_2^A < p_1$, consider the following transformation. Decrease $p_1$ until it reaches $p_2^R$ or $p_{avg} = \frac{p_1 + p_2^A}{2}$ (i.e., $p_1 = \max\{p_2^A, p_{avg}\}$). Increase $p_2^A$ by the same amount. The seller still obtains the same revenue from sophisticated buyers because $p_1$ never goes below $p_2^R$ and $p_{avg}$ remains the same. The seller obtains the same revenue if a na\"ive buyer rejects $p_1$. If the na\"ive buyer accepts $p_1$, we obtain weakly more revenue. To see this, let $\hat{p_1}$ and $\hat{p_2^A}$ be the updated prices. Let $R$ be the revenue curve. The seller previously obtained revenue $(p_1 + p_2^A)(1-F(p_1))$ from the na\"ive buyers. The seller now obtains revenue $(\hat{p}_1+\hat{p}_2^A)(1-F(\hat{p}_1))$ from the na\"ive buyer. Observe that $p_1 + p_2^A = \hat{p}_1+\hat{p}_2^A$ because $p_{avg}$ does not change. However, observe that $1-F(\hat{p}_1) \geq 1-F(p_1)$, so weakly more revenue. Thus, in this case we have shown a new equilibrium where $\hat{p}_1 \leq \hat{p}_2^A$.
\end{proof}
\begin{theorem}
\label{thm:priceorder}
There exists a revenue-optimal price schedule such that $p_2^R \leq p_1 \leq p_2^A$.
\end{theorem}
\begin{proof}
We know from Lemma~\ref{lem:p1&p2a} that there exists an equilibrium such that $p_1 \leq p_2^A$. We want to show that it is also possible that $p_2^R \leq p_1$ in this equilibrium. We proceed by contradiction. Assume that this is not the case, i.e., in all equilibria such that $p_1 \leq p_2^A$, it must be the case that $p_2^R > p_1$. Consider a transformation where $p_2^R$ is decreased to $p_1$. The seller still receives the same revenue from na\"ive agents because if a na\"ive agent rejected $p_1$, they would have definitely reject $p_2^R$ previously, and if the na\"ive agent accepted $p_1$, they would never be offered price $p_2^R$. Therefore, decreasing $p_2^R$, cannot decrease the revenue. The seller also receives the same revenue from the sophisticated agents. A sophisticated agent may be interested in 0, 1, or 2 items. If the sophisticated agent previously wanted 0 items, they will still reject $p_2^R$ after it has been decreased to $p_1$. If the sophisticated agent previously wanted 1 item, the revenue generated for the seller will still be $p_1$. If the sophisticated agent previously wanted 2 items, they will still accept $p_1$ after $p_2^R$ has been decreased to $p_1$. Thus, for both agent types, the seller obtains the same revenue from decreasing $p_2^R$ to $p_1$. Thus we have shown the existence of an equilibrium such that $p_2^R \leq p_1 \leq p_2^A$.
\end{proof}


We will now attempt to determine the maximum revenue that can be generated with respect to $\mu$ for the optimal price schedule. Observe that the revenue generated is the following:
\begin{align*}
    \textsc{Rev}(\mu) = & \mu(p_2^R(F(t)-F(p_2^R))+(p_1+p_2^A)(1-F(t)))\\
    & + (1-\mu)(p_2^R(F(p_1)-F(p_2^R))+p_1(1-F(p_1))+p_2^A(1-F(p_2^A)))
\end{align*}
where $t=p_1+p_2^A-p_2^R$ (this follows from the fact that $t-p_2^R=2t-p_1+p_2^A$ which is when the utility of one item equals the utility of two items).

\begin{theorem}
The revenue generated in the seller's optimal price schedule decreases as $\mu$ increases.
\end{theorem}
\begin{proof}
We will break the generated revenue into different cases, and show in all cases that the revenue generated by na\"ive agents is equal or greater than that of the sophisticated agents. Thus, this will imply that decreasing $\mu$ can only increase the revenue. We will assume we are in an equilibrium where $p_2^R \leq p_1 \leq p_2^A$ which is implied by Theorem~\ref{thm:priceorder}. If the value of either type of buyer is below $p_2^R$, the revenue generated will be 0 revenue. If the value of either type of buyer is below $p_1$, the revenue generated will be $p_2^R \cdot \Pr[p_2^R \leq v < p_1]$. If the value of either type of buyer is above $t$, the revenue generated will be $(p_1+p_2^A) \cdot \Pr[v > t]$. If the value of either type of buyer is between $p_1$ and $t$, the revenue generated by the different types of buyers will be different. In the case of the na\"ive buyer, we find that the revenue is $p_1 \cdot \Pr[p_1 \leq v \leq t]$, and in the case of the sophisticated buyer, we find that the revenue is $p_2^R \cdot \Pr[p_1 \leq v \leq t]$. Observe that the revenue of the na\"ive buyer is greater; therefore, the revenue increases as $\mu$ is decreased because we are placing more weight on the revenue of the na\"ive buyer.
 \end{proof}

\section{No-Learning Equilibrium}
\label{app:nolearning}

This appendix contains the strategies for the no-learning equilibrium of \citet{dps19} for the infinite-horizon setting.
The belief updates are such that if a buyer ever takes an off-path action (accepting $p>0$ or rejecting $p=0$), the beliefs update to a pointmass on $1$.
Otherwise, beliefs remain unchanged from round to round.
The precise description of the strategies we use is borrowed from \citet{ilpt17}.


\begin{algorithm}[H]
\label{alg:0RESell}
\SetKwInOut{Input}{Input}\SetKwInOut{Output}{Output}
\Input{Purchasing history $h^k$, belief support $[a,b]$.}
\Output{Price $p_{k+1}$}
\BlankLine
\uIf{Buyer has ever accepted a positive price}{$p_{k+1}=1$\;}
\uElseIf{
Buyer has ever rejected a price of $0$
}{$p_{k+1}=1$\;}
\Else{$p_{k+1}=0$\;}
\caption{Zero-Revenue Equilibrium - Seller's Strategy}
\end{algorithm}

\begin{algorithm}[H]
\label{alg:0REBuy}
\SetKwInOut{Input}{Input}\SetKwInOut{Output}{Output}
\Input{Purchasing history $h^k$, belief support $[a,b]$, value $v$, price $p_{k+1}$}
\Output{Purchasing decision for round $k+1$}
\BlankLine
\uIf{$p_{k+1}=0$}{
Accept\;}
\ElseIf{$p_{k+1}>0$}{
\uIf{Buyer has ever accepted a positive price}{Accept\;}
\uElseIf{Buyer has ever rejected a price of $0$}{Accept\;}
\Else{Reject\;}}
\caption{Zero-Revenue Equilibrium - Buyer's Strategy}
\end{algorithm}

\section{Full description of infinite horizon equilibrium with 3-valued support}
\label{app:infinite3sup}

This appendix gives a complete presentation of the buyer and seller strategies for the infinite-horizon equilibrium summarized in Section~\ref{sec:infeq}.
We consider a discrete distribution: we assume that both naive and sophisticated values are uniformly distributed over $V=\{1,10,20\}$. 
The mass $(1-\sophprob)$ of naive buyers is some small $\epsilon$, and we take the discount factor $\delta=2/3$. 

Equilibrium requires us to specify the buyer and seller strategies for every possible history of play, including those off the equilibrium path.
Play on the equilibrium path is already described in Section~\ref{sec:infeq}, and so we focus on giving the complete and well-organized description at all histories of play here.
The equilibrium we present has the feature that it is {\em Markovian}, meaning that the strategies only depend on the current beliefs, and the full history is not required to compute them.
Consequently, we may instead summarize histories by the beliefs they induce (given the strategies), then define the strategies as a function of the current beliefs.
This allows us to reduce the casework somewhat.
To this end, $N_t=\{\ell^N_t,\dots,\h^N_t\}$ be the updated support of naive buyer types according to the period-$t$ beliefs.
If $N_t=\emptyset$ we say that $h_t$ is \emph{inconsistent} with naive types.
Similarly, we use the sophisticated buyer's equilibrium strategy to define the updated support of sophisticated buyer types  $S_t=\{\ell^S_t,\dots,\h^S\}$. Since our buyer's strategy is deterministic, the tuple $(N_t,S_t)$ fully determines the seller's beliefs after history $h^t$; belief updates do not change the relative probabilities of the points remaining in the support.
Thus, we will describe the strategies of the seller and the sophisticated buyer as a function of $(N_t,S_t)$.

\begin{algorithm}
\DontPrintSemicolon
\KwIn{History: $h_t$}
\KwOut{Price $p$.}\;
Let $N_t=\{\ell^N_t,\dots,\h^N_t\}$ be the naive support according to $h_t$\;
Let $S_t=\{\ell^S_t,\dots,\h^S_t\}$ be the sophisticated support $h_t$\;
\;

\If{$N_t=\emptyset$ \tcp*{History inconsistent with naive. Only possible if $S_t=\{\ell^S_t\}$ 
}}{Set $p=\ell^S_t$ \tcp*{Offer price equal to only remaining sophisticated type}
}   
\ElseIf{$N_t=\{x\}$ }{Set $p=x$ \tcp*{Offer price equal to only remaining naive type}} 
\ElseIf{$N_t=\{10,20\}$ \tcp*{Either $S_t=\emptyset$ or $S_t=\{20\}$} 
}
{
\If{$S_t=\emptyset$   } {Set $p=20$ \tcp*{Optimal price for naive types}}
\Else{Set $p=10$ \tcp*{Sophisticated types reduce optimal price}}
}
\ElseIf{$N_t=\{1,10\}$ or $\{1,10,20\}$ \tcp*{Only possible if $S_t=\{1,10,20\}$} 
}{Set $p=2$ \tcp*{Optimal round $1$ price}}
\caption{{\sc Seller Strategy}}
\label{alg:seller3}
\end{algorithm}

\begin{algorithm}
\DontPrintSemicolon
\KwIn{History: $h_t$, value $v$}
\KwOut{Decision}\;
Let $N_t=\{\ell^N_t,\dots,\h^N_t\}$ be the naive support according to $h_t$\;
Let $S_t=\{\ell^S_t,\dots,\h^S_t\}$ be the sophisticated support $h_t$\;
\;

\If{$N_t=\emptyset$, $p\leq \ell^N_t$, or $N_t=\{20\}$  }{Accept if $v\geq p$
\tcp*{Buyer acts as price-taker}
}  
 
\ElseIf{$N_t=\{10\}$ or $\{10,20\}$ and $p>10$\tcp*{Either $S_t=\emptyset$ or $S_t=\{20\}$} 
}{Accept if $v\geq p+20$  
\tcp*{No type in $S_t$ accepts}
}
\ElseIf{$N_t=\{1,10,20\}$ and $p>1$ \tcp*{Only possible if $S_t=\{1,10,20\}$}}
{
{
\If{$p>10$  }{Accept if $v\geq p+\tfrac{26}{3}$ \tcp*{No type in $S_t$ accepts
}}
\If{$p\in (2,10]$}{Accept if $v\geq p+38$ \tcp*{No type in $S_t$ accepts 
}}

\ElseIf{$p\in (1,2]$}{Accept if $v\geq p+18$ \tcp*{Only type $v=20$ accepts
}}
}
}
\ElseIf{$N_t=\{1,10\}$ and $p>1$ \tcp*{Only possible if $S_t=\{1,10,20\}$}}
{
{
\If{$p>10$
}{Accept if $v\geq p+\tfrac{26}{3}$ \tcp*{No type in $S_t$ accepts}}
\ElseIf{$p\in (1,10]$}{Accept if $v\geq p+18$ \tcp*{Only $v=20$ accepts $p\in (1,2]$}}
}
}

\ElseIf{$N_t=\{1\}$ and $p>1$ \tcp*{Either $S_t=\{1,10,20\}$ or $S_t=\{1,10\}$} }
{
\If{$S_t=\{1,10,20\}$}{Accept if $v\geq p+38$\tcp*{No type in $S_t$ accepts}}
\ElseIf{$S_t=\{1,10\}$}{Accept if $v\geq p+18$\tcp*{No type in $S_t$ accepts}}
}

\caption{{\sc Buyer Strategy}}
\label{alg:buyer3}
\end{algorithm}

The seller's strategy can be summarized as follows: if there are zero naive types, the seller posts the bottom of the sophisticated support.
If there is one naive type remaining, the seller offers this type for the remainder of the game. (Strategies are such that this naive type will be below the sophisticated support, always.)
If there are multiple naive types and no sophisticated types, the seller engages in optimal price search over these buyers.
Otherwise, the seller offers a price of $2$, which separates the sophisticated types as described in Section~\ref{sec:infeq}. Full details can be found in Algorithm~\ref{alg:seller3}.

The buyer's strategy can be summarized as follows: The buyer with value $1$ is always a price taker. A buyer with value $10$ accepts only $p\leq 1$ unless the history has already eliminated $v=1$ from the naive support $N$, in which case the agent becomes a price taker. Buyer with value $v=20$ considers four cases: If $\ell^N_s=10,20$, then accepts if $p\leq l^N_s$. If both values $1$ and $10$ are in the naive support then the buyer accepts if $p\leq 2$. Otherwise, if the only value in the naive support is $N_t=\{1\}$, the buyer accepts if $p\leq 1$. Full details can be found in Algorithm~\ref{alg:buyer3}

\subsection*{Belief Updates}
Given history $h_t$ or equivalently $(N_t,S_t)$, price offered $p$ and a buyer decision, the seller needs to update their beliefs: $N_{t+1},S_{t+1}$ 
The beliefs will be updated using the Bayes rule whenever applicable; we omit an explicit description of these updates for brevity.
Below, we explicitly detail the updates for off-path actions, where all types currently supported in the beliefs are supposed to take the opposite action of what was observed.
In this case, we update beliefs to full sophistication ($N_{t+1}=\emptyset$), and update the sophisticated distribution to a pointmass as follows. Define below that $\h_t^S$ denotes the sophisticated upper bound and $\ell_t^S$ the upper bound at time $t$. Then:
\begin{itemize}
    \item If $S_t\neq\emptyset$:
   \begin{itemize}
       \item     If the buyer accepted $p$ but they shouldn't have set $S_{t+1}=\{\h^S_t\}$.
    \item If the buyer rejected  $p$, but they shouldn't accept set $S_{t+1}=\{\ell^S_t\}$.
   \end{itemize} 
   \item $S_t=\emptyset$: Set $S_{t+1}=\{20\}$.
\end{itemize}
Other belief updates are possible here; these were selected largely for simplicity.

\subsection*{Seller's Incentives}
The casework below shows that the seller is best responding at every reachable belief by playing the strategy prescribed in Algorithm~\ref{alg:seller3}.\\

\nind $\mathbf{N_t=\emptyset}$ \\ 
These beliefs only occur if the buyer took an off-equilibrium path action.
The seller has updated their beliefs to some $S_{t}=\{\ell^S_t\}=\{\h^S_t\}$ at some point. Since every buyer acts as price-taker in this event and setting $p=\ell^S_t$ is optimal.\\
    
    \nind $\mathbf{N_t=\{x\}}$ \\ Consider a price $p>x$, we argue that either $S_t=\emptyset$ or no type in $S_t$ accepts $p$.
    In either case, pricing $p=x$ is optimal.
    \begin{itemize}
        \item $N_t=\{20\}$.  $S_t=\emptyset$ since some $p'\in (10,20]$ was accepted previously. 
        \item $N_t=\{10\}$. A sophisticated type accepts if $v\geq p+20>20$.
        \item $N_t=\{1\}$ and $S_t=\{1,10,20\}$. A sophisticated agent accepts if $v>38+p>39$ .
        \item $N_t=\{1\}$ and $S_t=\{1,10\}$. A sophisticated agent accepts if $v\geq 18+p>19$.
        \end{itemize}
    Thus, either $S_t=\emptyset$ or all types in $S_t$ reject the price $p>x$. On the other hand every price $p\leq x$ will be accepted: the buyer acts as price-taker, and since $\ell^N_t =x \leq l^S_t$, every type in $S_t$ (if not-empty) will accept. \\
         
    \nind $\mathbf {N_t=\{10,20\}}$ and $\mathbf{S_t=\emptyset}$ \\  The price $p=20$ provides the optimal expected revenue, as $1/2\cdot 20\tfrac{1}{1-\delta}+ 1/2\cdot 10 \tfrac{\delta}{1-\delta}=40$. \\

    \nind  $\mathbf{N_t=\{10,20\}}$, and $\mathbf{S_t\neq\emptyset}$ \\
    Note that this is only possible if $S_t=\{20\}$: since $N_t=\{10,20\}$ then some price $p'\in (1,10]$ has been offered. In any naive-consistent history, sophisticated types $v=1,10$ must reject this price, and hence, $1,10\notin S_t$. 
    
    Now consider any price $p>10$. Since $20<p+20=t$, the sophisticated agent rejects. We just need to consider the revenue for $p=20$. The price is accepted with probability $\frac{1-\mu}{2}$. If the buyer accepts the seller offers $20$ forever, and if the buyer rejects the sellers offers $10$ forever. 
    The total revenue is $$ 20\frac{1-\mu}{2}\frac{1}{1-\delta}+ 10\frac{1+\mu}{2}\frac{\delta}{1-\delta}  =40-20\mu$$
    Price $p=10$ will always be accepted by all agents for total revenue $10\frac{1}{1-\delta}=30$. Thus if $\mu\geq 1/2$ the seller's optimal price is $p=10$.\\

    \nind $\mathbf{N_t=\{1,10\}}$ or $\mathbf{N_t=\{1,10,20\}}$ \\
    Either no previous price revealed information about the naive buyers, or a price was rejected in $(10,20]$. Since this is consistent with all sophisticated types we have that $S_t=\{1,10,20\}$. 
    
    Offering price $p=2$ is accepted by naive agents with values $10$ and $20$ (if available)  and sophisticated agents with value $v\geq 20$. The seller will keep offering $10$ after an accept, or $1$ after a reject. Since naive agents provide more revenue than sophisticated agents per capita, we can calculate a lower bound on discounted revenue by assuming $\mu=1$ (which holds for both cases, $N_t=\{1,10\}$ and $N_t=\{1,10,20\}$).
    \begin{equation*}
	   \tfrac{1}{3}\cdot(2+\tfrac{\delta}{1-\delta}\cdot 10)+\tfrac{2}         {3}\cdot\tfrac{\delta}{1-\delta}\cdot 1=\tfrac{26}{3}
    \end{equation*}
    We can compare the equilibrium revenue with other possible prices.
    \begin{itemize}
    \item Prices $p\in (1,2)$ is dominated by $p=2$, since the same types accept both prices.
    \item Prices $p>2$ are rejected by all sophisticated agents; all future prices will be equal to $1$ since we updated $N_t=\{1\}$. Total sophisticated revenue is $\frac{\delta}{1-\delta}=2$. We can upper bound the revenue contribution of naive agents by $20\frac{1}{1-\delta}=60$. Thus, the revenue of any price is at most $2\mu+(1-\mu)60$.
    \item If price  $p=1$ was optimal it would imply that posting $p_1$ forever is optimal for total revenue of $\frac{1}{1-\delta}=3$.
    
    \end{itemize}
    Clearly, the sophisticated revenue is uniquely maximized for $p=2$. Thus, for small $\epsilon =(1-\mu)$, setting $p=2$ is optimal.

\subsection*{Buyer Incentives}

Before proceeding with the case analysis, we observe some useful facts about the equilibrium:
\begin{enumerate}
    \item The equilibrium prices are always in $[\ell^N_t,\h^N_t]$ if $N_t\neq\emptyset$.
    \item If $N_t\neq \emptyset$ we have that $\ell^N_t\leq \ell^S_t$.
    \item The buyer utility for $N_t\neq\emptyset$ is at least   $\tfrac{\delta}{(1-\delta)}(v-\ell^N_t)=2(v-\ell^N_t)$:
  \begin{itemize}
      \item If $|N_t|=1$, the seller offers $p=\ell^N_k$ always and if the buyer accepts it they get $\frac{1}{1-\delta}=3(v-\ell^N_k)$.
      \item If $|N_t|= 2$ and the seller offers $p\in (\ell^N_t,\h^N_t]$, the buyer can reject $p$ and accept all future prices.  Since the seller sets $N_{t+1}=\{\ell^N_t\}$, all future pries are $p=\ell^N_t$. The total discounted utility is $\frac{\delta}{1-\delta}=2(v-\ell^N_k)$.
      \item If $N_t=\{1,10,20\}$; the seller offers $p=2$, rejecting the price guarantees utility $2(v-\ell^N_k)=2(v-1)$.
  \end{itemize}   
  \item If accepting a price of $p>\ell^N_t$ would raise all future prices to $p'\geq v$, then the sophisticated type $v$ prefers to reject $p$: the accept utility is $(v-p)$, whereas the buyer can reject and obtain $\delta 2(v-\ell^N_t)=\tfrac{4}{3}(v-\ell^N_t)>(v-p)$.
    \end{enumerate}
    We complete the analysis of the example equilibrium with case analysis showing that the buyer strategies best respond to the seller's strategy in equilibrium. \\




\nind $\mathbf{p\leq \ell^N_t}$

\nind Then we have the following sub-cases.
    \begin{itemize}
        \item If $\ell^N_t\geq 10$, then either  $S_t=\emptyset$ or $S_t=\{20\}$. If the price  $p<l^N_t$ is rejected the seller will only offer $p=20$ in the future. Thus, all types act as price-takers. 
        \item If $\ell^N_t=1$, it must be that $\ell^S_t=1$. Thus, if a price $p<1$ is rejected, the seller will offer $p'=1$ all future rounds. The total utility in this case is $(v-1)\frac{\delta}{1-\delta}=2v-2$. Note that if $\ell^N_t=1$, then if ever a price is rejected all future prices will be equal to $p'=1$. Thus, the buyer can always achieve a better payoff by accepting $p\leq 1$, which doesn't update the beliefs, then reject the next price. The total utility is $(v-p)+(v-1)\frac{\delta^2}{1-\delta}=\frac{7}{3}v-1-p\geq 2v-2$. Thus, if $v\geq p$ the sophisticated agent prefers to accept. 
  \\
    \end{itemize}

\nind$\mathbf{  N_t=\emptyset}$\\    For this to occur it must be that $|S_t|=1$, so the future prices are independent of the buyer's actions. Therefore price taking is optimal.\\ 
    
\nind $\mathbf{N_t=\{20\}}$ \\ Note that $S_t=\emptyset$; the seller will only offer $20$ in the future regardless of the buyer's actions. Thus any type acts as price-taker.\\

\nind $\mathbf{N_t=\{10\}}$ or $\mathbf{N_t=\{10,20\}}$ \\ Note that $S_t=\{20\}$ or $S_t=\emptyset$.
Accepting $p>10$ is either an off-path action or sets $N_{t+1}=\{20\}$. All future prices will be $20$. According to our observation, any type $v\in (1,20)$ prefers to reject. For any type $v>20$, the accept utility is $(v-p)+2(v-20)$, and the reject utility is $2(v-10)$. Thus, any agent accepts if and only if they have value $v\geq p+10>20$. \\

\nind $\mathbf{N_t=\{1,10,20\}}$ \\
Note that  $S_t=\{1,10,20\}$. This is the most complicated case, so we break up the analysis based on the range of prices.
\begin{itemize}
\item $\mathbf{p>10}$ \\
Accepting a price $p>10$ either sets $N_{t+1}=\{20\}$ and $S_{t+1}=\emptyset$ or sets $N_{t+1}=\emptyset$ and $S_{t+1}=\{20\}$ (if accepting is off-path). In both cases, all future prices will be $20$. Rejecting a price $p$  will set $N_{t+1}=\{1,10\}$; Thus, any type in $v\leq 20$ prefers to reject. For $v>20$, the accept utility is $(v-p)+2(v-20)=3v-p-40$. After rejecting the seller will offer prices similar to the round $1$ equilibrium; any $v>20$ obtains utility 
$\delta ((v-2)+2(v-10)=2v-\tfrac{44}{3}$.
Thus, the buyer wants to accept if and only if  $v\geq p+\tfrac{26}{3}>20$.

\item $\mathbf{p\in (2,10]}$\\
Accepting a price $p\in (2,10]$ sets $N_{t+1}=\{10,20\}$ and $S_{t+1}=\emptyset$. In the next round the seller will offer $20$ as long as it is accepted. If the seller observes a reject they will offer price $10$ forever. 
Clearly, all types $v\leq 10$ prefer to reject. For types $v>10$, we have the following options:
\begin{enumerate}
    \item  Accept the price $p$; reject the price  $20$, and accept price $10$ for all future rounds. The discounted utility is $v-p + \tfrac{4}{3}(v-10)$.
    \item Accept price $p$; forever. The discounted utility  is  $v-p + 2(v-20)$.
    \item Reject price $p$; accept price $1$ forever. The discounted utility  is  $2(v-1)$.
\end{enumerate}
Type $v=p+38$ is indifferent between $(3)$ and  $(2)$. Type $v=3p+28$ is indifferent between $(3)$ and $(1)$. Thus, the agent accepts if and only if $v\geq \max \{3p+34,p+38\}=p+38>20$ for $p>2$.

\item $\mathbf{p\in (1,2]}$\\
Accepting a price $p\in (1,2]$  sets $N_{t+1}=\{10,20\}$ and $S_{t+1}=\{20\}$, therefore all future prices will be $10$. Thus, any type $v\leq 10$ prefers to reject. For value $v>10$, the discounted accept utility is $(v-p)+2(v-10)$;  the discounted reject utility is $2(v-1)$. Thus, the buyer accepts if and only if $v\geq p+18>10$.\\
\end{itemize}

\nind $\mathbf{N_t=\{1,10\}}$ \\
Note that  $S_t=\{1,10,20\}$. 
\begin{itemize}
\item $\mathbf{p\in (1,10]}$\\
Accepting a price $p\in (1,10]$ sets $N_{t+1}=\{10\}$; all future prices will be $10$. Rejecting the price sets $N_{t+1}=\{1\}$; all future prices will be $1$. Similarly to before, the buyer accepts if and only if $v\geq p+18>10$.
\item $\mathbf{p>10}$\\
Accepting a price $p>10$ is off-path and sets $S_{t+1}=\{20\}$ and $N_{t+1}=\emptyset$; all future prices will be $20$. Rejecting a price $p$  will set $N_{t+1}=\{1\}$;  all future prices will be $1$. As we argued before, the buyer accepts if and only if $v\geq p+38>20$. \\
\end{itemize}

\nind$\mathbf{N_t=\{1\}}$ 
\begin{itemize}
    \item  $S_t=\{1,10,20\}$\\ Accepting $p>1$ is an off-path action that results in $N_{t+1}=\emptyset$ and $S_{t+1}=\{20\}$. All future prices will $20$.
   After rejecting the seller will offer prices similar to the round 1 equilibrium.  This is similar to earlier case: the buyer accepts if and only if $v\geq p+\tfrac{26}{3}>20$.
    
    \item  $S_t=\{1,10\}$ \\ Accepting $p>1$ is an off-path action that results in $N_{t+1}=\emptyset$ and $S_{t+1}=\{10\}$. All future prices will $10$. 
    If the buyer rejects, all future prices will be $1$. As we argued before, the buyer wants to accept if and only if $v\geq p+18>10$. 
\end{itemize}

\section{Proof of Lemma~\ref{lem:discreterevenue}}
\label{app:discreteinf}

	Let $p^*_S$ denote the monopoly price for $F^S$, maximizing $p(1-F^S_+(p))$.
	Consider the support point just below $(1-\delta)p^*_S$, given by $\sup\{p\in V~|~p\leq (1-\delta)p^*_S\}$.
	Considering such a $p$, we have:
	\begin{align}
		\max_{p\in V}p(1-F( p/(1-\delta)))&\geq p(1-F(p/(1-\delta)))\notag\\
		&\geq (1-\delta)p^*_S(1-F(p/(1-\delta)))-\Delta\label{eq:delta1}\\
		&\geq (1-\delta)p^*_S\prob[v>p/(1-\delta)]-\Delta\label{eq:delta2}\\
		&\geq (1-\delta)p^*_S\prob[v\geq p^*_S]-\Delta\notag\\
		&\geq (1-\delta)p^*_S\prob[\mathcal E_S]\prob[v\geq p^*_S~|~\mathcal E_S]-\Delta\notag\\
		&=(1-\delta)\sophprob p^*_S(1-F^S_+(p^*))-\Delta.\notag
	\end{align}
Lines (\ref{eq:delta1}) and (\ref{eq:delta2}) follow from the fact that $p$ is the support point just below $p^*$.
The rest follow from definitions or basic probability.

\section{Discussion of Equilibrium with Zero Naive Buyers}
\label{app:zeronaive}

This appendix briefly discusses the sense in which the strategies we describe in Section~\ref{sec:infeq} and Appendix~\ref{app:infinite3sup} are still an equilibrium even with the absence of naive buyers.
The standard method of describing a PBE in this game is to specify actions in terms of histories of play; that is, in terms of past prices and accept decisions.
To facilitate the extensive casework, the exposition of Appendix~\ref{app:infinite3sup} deviates from this convention and largely describes strategies in terms of the beliefs.
Nonetheless, it is possible to take this description and obtain the more conventional mapping from histories of play to actions.
The claim is that this mapping from histories to actions continues to be a PBE in the absence of naive buyers.

This claim can be verified in broad strokes in the following way: in the incentive analysis for both the seller and the sophisticated buyers, the distribution or fraction of naive buyers never came up except to note that as long as this latter fraction was sufficiently small, agents were best responding. 
Consequently, all these arguments hold even when the naive fraction is zero.
The only exception to this observation is the case where the path of play would eliminate sophisticated buyers entirely.
This occurs, for example, if $p_1 = 3$ is accepted.
With all naive buyers eliminated, this is now off-path, and requires us to set beliefs for the sophisticated buyers.
Moreover, when the naive buyers were present, they were uniformly distributed on $\{10,20\}$ at this history, so it was optimal for the seller to engage in price exploration beyond this point.
Rather than try to implement this price exploration with only sophisticated buyers, we observe that we may set the sophisticated beliefs to be a point mass on $20$ and have the seller offer a price of $20$ for the remainder of the game.
The buyer's best responses with this change remain identical.
Furthermore, there are alternative naive buyer distributions where this fixed price of $20$ was optimal even when naive buyers were present.
Hence, we can make the stronger statement that there exists a naive buyer distribution such that the equilibrium with naive buyers is still a PBE {\em exactly}, even with only sophisticated buyers.

\end{document}